\documentclass[draftcls,onecolumn,12pt]{IEEEtran}
\usepackage{amssymb,amsmath}
\usepackage{amsmath}
\usepackage{amsfonts}

\newtheorem{proof}{Proof}

\usepackage{amssymb}
\usepackage[dvips,final]{epsfig}
\usepackage{amsfonts}
\usepackage{latexsym}

\usepackage{subfigure}
\usepackage{color}

\usepackage{setspace}

\newcommand{\e}{\mbox{$\mathbb E$}}

\newcommand{\oo}{\mbox{$\mathbb O$}}

\newcommand{\x}{\mbox{$\mathbf x$}}
\newcommand{\y}{\mbox{$\mathbf y$}}
\newcommand{\vv}{\mbox{$\mathbf v$}}

\newcommand{\aaa}{\mbox{$\cal A$}}

\newcommand{\f}{\mbox{$\cal F$}}

\newcommand{\p}{\mbox{$\cal P$}}
\newcommand{\q}{\mbox{$\cal Q$}}

\newcommand{\rrr}{\mbox{$\cal R$}}

\newcommand{\po}{\mbox{$\it\Phi$}}
\newcommand{\up}{\mbox{$\it\Upsilon$}}

\newcommand{\qa}{\mbox{\quad\mbox{and}\quad}}

\newcommand{\bfF}{\mbox{\boldmath $F$}}

\newcommand{\bfX}{\mbox{\boldmath $X$}}

\newcommand{\bfY}{\mbox{\boldmath $Y$}}

\newcommand{\rt}{\mbox{$\mathbb R$}}

 \def\diag{\mathop{{\rm diag}}\nolimits}

\newtheorem{theorem}{Theorem}

\newtheorem{corollary}{Corollary}

\newtheorem{lemma}{Lemma}

\newtheorem{example}{Example}

\newtheorem{remark}{Remark}

\begin{document}
%
\title{Compression and Recovery of Distributed Random Signals}
%
%
%

\author{ Alex~Grant, Anatoli~Torokhti,
        and~Pablo~Soto-Quiros
\thanks{Alex Grant was with the Institute for Telecommunications Research, University of South Australia, SA 5095, Australia.}
\thanks{Anatoli Torokhti is with the Centre for Industrial and Applied Mathematics, University of South Australia, SA 5095, Australia (e-mail: anatoli.torokhti@unisa.edu.au).}
\thanks{Pablo Soto-Quiros is with the Centre for Industrial and Applied Mathematics, University of South Australia, SA 5095, Australia and the Instituto Tecnologico de Costa Rica, Apdo. 159-7050, Cartago, Costa Rica (e-mail: juan.soto-quiros@mymail.unisa.edu.au).}
\thanks{Manuscript received XXXX XX, 2015; revised XXXX XX, 2015.}}


\maketitle

\begin{abstract}
We consider the case when a set of spatially distributed sensors  $\mathcal{Q}_1, \ldots,\mathcal{Q}_p$ make local  observations, ${\bf y}_1,\ldots,{\bf y}_p$, which are noisy versions of a signal of interest, ${\bf x}$. Each sensor  $\mathcal{Q}_j$ transmits compressed information ${\bf u}_j$ about its measurements to the fusion center which should recover the original signal within a prescribed accuracy. Such an information processing relates to a wireless sensor network (WSN) scenario.
 The key problem is to find models of the sensors and fusion center so that they  will be optimal in the sense of minimization
of the associated error under a certain criterion, such as the mean square error (MSE).
We determine the models from the technique which is  a combination of the  maximum block improvement  (MBI) method \cite{chen87,Zhening2015} and  the generic   Karhunen-Lo\`{e}ve transform (KLT) \cite{torbook2007} (based on the work in \cite{Torokhti2007,tor5277}). Therefore, the proposed method unites  the merits of both techniques   \cite{chen87,Zhening2015} and \cite{torbook2007,Torokhti2007,tor5277}. As a result, our approach provides, in particular,
the minimal MSE at each step of the version of the MBI method we use. The WSN model is represented in the form called the  multi-compressor  KLT-MBI transform.  The multi-compressor  KLT-MBI is given in terms of pseudo-inverse matrices and, therefore, it is
numerically stable and always exists. In other words, the proposed WSN  model provides compression, de-noising and reconstruction of distributed signals for the cases  when known  methods either are not applicable (because of singularity of associated matrices)  or produce  larger associated errors. Error analysis is provided.
\end{abstract}


\IEEEpeerreviewmaketitle

\section{Introduction}

\subsection{Motivation}
   \IEEEPARstart{W}{e}  seek to find effective numerical algorithms for an information processing
scenario that involves a set of spatially distributed sensors, $\mathcal{Q}_1,\ldots,\mathcal{Q}_p$, and a fusion
center, $\p$. The sensors make local  observations, ${\bf y}_1,\ldots,{\bf y}_p$, which are noisy versions of a signal
of interest, ${\bf x}$. Each sensor  $\mathcal{Q}_j$ transmits compressed information ${\bf u}_j$ about its
measurements to the fusion center which should recover the original signal
within a prescribed accuracy. Such an information processing relates to a wireless sensor
network (WSN) scenario. In the recent years, research and development on new and
refined WSN techniques has increased at a remarkable rate (see, for example, \cite{dragotti2009,5447742,5504834,Bert2010_1,6334203,ma3098,Marelli201527}). This is, in particular, because
of a multitude of WSN applications due to their low deployment and maintenance cost.

\begin{figure}
  \centering
  \includegraphics[width=9cm]{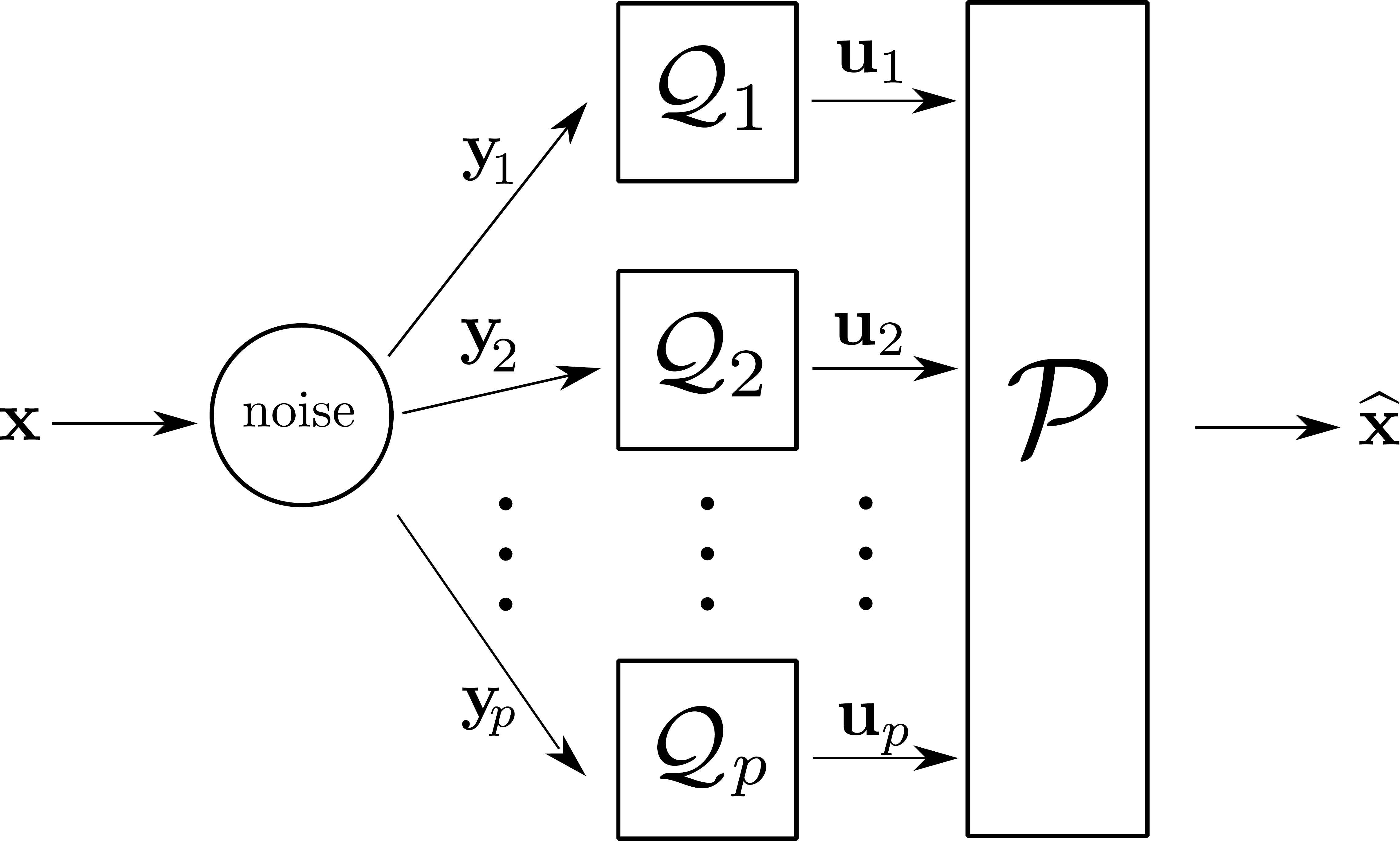}\\
  \caption{Block diagram of the WSN. Here,  $\widehat{\bf x}$ is an estimation of ${\bf x}$.}
  \label{fig1}
\end{figure}

The key problem is to to find an effective way to {\it compress} and denoise each observation
${\bf y}_j$, where $j = 1,\ldots, p$, and then {\it reconstruct} all the compressed observations in the
fusion center so that the reconstruction will be optimal in the sense of a minimization
of the associated error under a certain criterion, such as the mean square error (MSE).
A restriction is that the sensors cannot communicate with each other. Here, the term
``compression'' is treated in the same sense as in the known works on data compression
(developed, for instance, in \cite{Scharf1991113,681430,905856,Torokhti20102822}), i.e. we say that observed signal ${\bf y}_j$ with $n_j$
components is compressed if it is represented as signal ${\bf u}_j$ with $r_j$ components where $r_j < n_j$,
for $j = 1,\ldots, p$. That is ``compression'' refers to dimensionality reduction and not quantization
which outputs bits for digital transmission. This is similar to the way considered, in particular, in \cite{4276987}.

\subsection{ Known techniques} \label{known}

 It is known that in the nondisrtibuted setting (in the other words,
in the case of a single sensor only) the MSE optimal solution is provided by the Karhunen-Lo\`{e}ve
transform (KLT) \cite{Scharf1991113,681430,2001309,torbook2007}. Nevertheless, the classical KLT cannot be applied to the above
WSN since the entire data vector ${\bf y} = [{\bf y}_1^T ,\ldots, y_p^T]^T$ is not observed by each sensor.
Therefore, several approaches to a determination of mathematical models for $\mathcal{Q}_1 ,\ldots, \mathcal{Q}_p$
and $\p$ have been pursued. In particular, in the information-theoretic context, distributed
compression has been considered in the landmark works of Slepian and Wolf \cite{1055037}, and
Wyner and Ziv \cite{1055508}. A transform-based approach to distributed compression and the
subsequent signal recovery has been considered in \cite{ma3098,4276987,Song20052131, 1420805, 4016296,  4475387}. Intimate relations
between these two perspectives have been shown in \cite{952802}. The methodology developed in
\cite{4276987,Song20052131, 1420805, 4016296,  4475387} is based on the dimensionality reduction by linear projections. Such an approach has received considerable attention (see, for example, \cite{dragotti2009, 5447742, 5504834,6334203, Schizas2007,Saghri2010}).

In this paper, we consider a further extension of the methodology studied in \cite{4276987,Song20052131, 1420805, 4016296,  4475387}.
In particular, in \cite{4276987}, two approaches are considered. By the first approach, the
fusion center model, $\p$, is given in the form $\p = [\p_1,\ldots,\p_p]$ where $\p_j$ is a `block' of
$\p$, for $j = 1,\ldots, p$, and then the original MSE cost function is represented as a sum
of $p$ decoupled MSE cost functions. Then approximations to $\mathcal{Q}_j$ and $\p_j$ are found as
solution to each of the $p$ associated MSE minimization (MSEM) problems. We observe
that the original MSE cost function and the sum of $p$ decoupled MSE cost functions
are not equivalent. This is because the covariance matrix cannot be represented in a
block diagonal form in the way presented in \cite{4276987}. Some more related explanations can
be found in \cite{Torokhti2012}. Therefore, the first approach in \cite{4276987} leads to the corresponding increase
in the associated error. The second approach in \cite{4276987} generalizes the results in \cite{Song20052131, 1420805}
in the following way. The original MSE cost function is represented, by re-grouping its
terms, in the form similar to that presented by a summand in the decoupled MSE cost
function. Then the minimum is seeking for each $\p_j\mathcal{Q}_j$ , for $j = 1,\ldots,p$, while other
terms $\p_k\mathcal{Q}_k$, for $k = 1,\ldots,j - 1, j + 1,\ldots, p$, are assumed to be fixed. The minimizer
follows from the known result given in \cite{Brillinger2001} (Theorem 10.2.4). While the original MSEM
problem is stated for a simultaneous determination of $\mathcal{Q}_1,\ldots,\mathcal{Q}_p$ and $\p$, the approach
in \cite{4276987} requires to solve $p$ local MSEM problems which are not equivalent to the original
problem. To combine solutions of those $p$ local MSEM problems, approximations to each
sensor model $\mathcal{Q}_j$ are determined from an iterative procedure. Values of $\mathcal{Q}_1,\ldots,\mathcal{Q}_p$ for
the initial iteration are chosen randomly.

The approach in \cite{4475387} is based on the ideas similar to those in the earlier references \cite{4276987,Song20052131, 1420805,4016296}\footnote{In particular,
it generalizes work \cite{4016296} to the case when the vectors of interest
do not to be directly observed at the sensors.}, i.e. on the replacement of the original MSEM problem  with the $p + 1$ unconstrained MSEM
problems for separate determination of approximations to $\mathcal{Q}_j$ , for each $j=1,\ldots, p$, and
then an approximation to $\p$. First, an approximation to each $\mathcal{Q}_j$ , for $j = 1,\ldots, p$, is
determined under assumption that other $p-1$ sensors are fixed. Then, on the basis of
known approximations to $\mathcal{Q}_1,\ldots,\mathcal{Q}_p$, an approximation of $\p$ is determined as the optimal
Wiener filter. Those $p + 1$ problems considered in \cite{4475387} are not equivalent to the original
problem. In \cite{4475387}, the involved signals are assumed to be zero-mean jointly Gaussian
random vectors. Here, this restriction is not used.

The work in \cite{Song20052131, 1420805,4016296} can be considered as a particular case of \cite{4475387}.

{ The method in \cite{ma3098} is applied to the problem which is an approximation of the original problem. It implies an increase in the associated error  compared to the method applied to the original problem. Further, it  is applicable under certain restrictions imposed on  observations and associated covariance matrices. In particular, in \cite{ma3098}, the observations should be  presented in the special form ${\bf y}_j = H_j \x + \vv_j$, for $j=1,\ldots,p$ (where $H_j$ is a measurement matrix and $\vv_j$ is noise),  and the covariance matrix formed by the noise vector should be {\em block-diagonal} and invertible. It is not the case here.
 }

\subsection{ Differences from known methods. Novelty and Contribution}\label{differences}

The WSN models in \cite{ma3098,4276987,Song20052131, 1420805, 4016296,  4475387} are  justified  in terms of inverse matrices.
It is known that in the cases when the matrices are close to singular this may lead to instability and
significant increase in the associated error. Moreover, when the matrices are singular,
the algorithms \cite{4276987,Song20052131, 1420805, 4016296, 4475387} may not be applicable.
This observation is illustrated by Examples \ref{ex1}, \ref{ex2} and  \ref{ex3} in Section \ref{simulations} below where  the associated matrices are singular and the method  \cite{4276987} is not applicable.
Although  in \cite{4475387}, for the case when a matrix is singular, it is proposed  to replace its inverse by the pseudo-inverse, such a simple  replacement
does not follow from the justification of the  model provided in \cite{4475387}. As a result, a simple substitution of the pseudo-inverse matrices instead of the inverse matrices may lead to the numerical instability as it is shown, in particular, in  Example \ref{ex3} in Section \ref{simulations}.
In this regard, we also refer to references \cite{torbook2007,Torokhti2007,Torokhti2009661} where the case of rank-constrained data compression in terms of the pseudo-inverse matrices  is studied.

Thus, methods in  \cite{4276987,Song20052131, 1420805, 4016296,  4475387,ma3098} are {\em justified} only for full rank matrices used in the associated models. This is not the case here. On the basis of the methodology developed in  \cite{torbook2007,Torokhti2007,Torokhti2009661}, our method is rigorously justified for models with matrices of degenerate ranks (Sections \ref{greedy} and \ref{appendix}).

{ Further, the second approach in \cite{4276987} is, in fact, the block coordinate descent (BCD) method \cite{Tseng2001}. The BCD method  converges (to a local minimum) under the assumptions that the objective function and the space of optimization need to be convex or the minimum of the objective function is {\em uniquely} attended (see \cite{Tseng2001} and \cite{bertsekas1995nonlinear}, p. 267). Those conditions are not satisfied for our method (as for the method  in \cite{4276987} as well).  Unlike the BCD method  \cite{Tseng2001} used in \cite{4276987}, the maximum block improvement  (MBI) method  \cite{chen87,Zhening2015} avoids the requirements of the BCD method. The MBI method guaranties convergence to a coordinate-wise minimum point which is a local minimum of the objective function.
Therefore, the approach proposed in this paper is based,  in particular, on the idea of the MBI method.

As distinct from the technique in \cite{ma3098} our method is applied to the original  minimization problem, not to  an approximation of the original problem as in \cite{ma3098}. It allows us to avoid the  increase in the associated error. We also do not impose any of the restrictions on our method as those in \cite{ma3098}  mentioned in Section \ref{known}  above. In particular, we do not assume that the covariance matrix formed by the noise vector should be  block-diagonal.
}

The methods in \cite{4276987,Song20052131, 1420805, 4016296, 4475387} have been developed under assumption that exact covariance matrices
are known. A knowledge of exact covariance matrices might be a restrictive condition in
some cases, in particular, when matrices are large.
In this paper, these difficulties are mitigated to some extent.

Key advantages of the proposed method are as follows. {The method represents a combination of the generic   Karhunen-Lo\`{e}ve
transform (KLT) \cite{torbook2007} (based on the work in \cite{tor5277}) and the MBI method \cite{chen87,Zhening2015}. Therefore, it unites  the merits of both techniques \cite{torbook2007,tor5277} and  \cite{chen87,Zhening2015}. As a result, our approach provides, in particular,
the minimal MSE at each step of the version of the MBI method we use. The WSN model is represented in the form called the {\em multi-compressor  KLT-MBI transform} and is based on the ideas which are different from those used in known methods \cite{ dragotti2009, 5447742, 5504834, 6334203,4276987,Song20052131, 1420805, 4016296,4475387,Schizas2007}.
}
 The multi-compressor  KLT-MBI is given in terms of pseudo-inverse matrices and, therefore, it is
numerically stable and always exists. In other words, the proposed WSN  model provides compression, de-noising and reconstruction of distributed signals for the cases  when known  methods either are not applicable (because of singularity of associated matrices)  or produce  larger associated errors. This observation is supported, in particular, by results of simulations in Section \ref{simulations}.


\subsection{Notation}\label{notation}

 Here, we provide some notation which is required to formalize the problem
in the form presented in Section \ref{stat1} below. Let us write $(\Omega,\Sigma,\mu)$ for a probability space\footnote{$\Omega=\{\omega\}$ is the set of outcomes, $\Sigma$ a $\sigma-$field of measurable subsets of $\Omega$ and $\mu:\Sigma\rightarrow[0,1]$ an associated probability measure on $\Sigma$ with $\mu(\Omega)=1$.}. We denote by ${\bf x} \in  L^2(\Omega,\mathbb{R}^m)$ the signal of interest\footnote{The space $L^2(\Omega,\mathbb{R}^m)$ has to be used because of the norm introduced below in (\ref{eq6}).} (a source signal to be estimated) represented as ${\bf x}=[{\bf x}^{(1)},\ldots,{\bf x}^{(m)}]^T$ where ${\bf x}^{(j)}\in L^2(\Omega,\mathbb{R})$, for $j=1,\ldots,m$. Futher, ${\bf y}_1 \in  L^2(\Omega,\mathbb{R}^{n_1}),\ldots,{\bf y}_p \in  L^2(\Omega,\mathbb{R}^{n_p})$ are observations made by the sensors where $n_1+\ldots + n_p = n$. In this regard, we write
\begin{equation}\label{eq1}
{\bf y}=[{\bf y}^{T}_1,\ldots,{\bf y}^{T}_p]^T\;\text{ and }\;{\bf y}=[{\bf y}^{(1)},\ldots,{\bf y}^{(n)}]^T
\end{equation}
where ${\bf y}^{(k)} \in L^2(\Omega,\mathbb{R})$, for $k = 1,\dots, n$. We would like to emphasize {\em a difference between
${\bf y}_j$ and $\y^{(k)}$:} in (\ref{eq1}), the observation ${\bf y}_j$, for $j = 1,\dots,p$, is a `piece' of random vector ${\bf y}$
(i.e. ${\bf y}_j$ is a random vector itself), and ${\bf y}^{(k)}$, for $k = 1,\ldots,n$ is an entry of y (i.e. ${\bf y}^{(k)}$ is
a random variable).\footnote{Therefore,   $\y_j = [\y^{(n_0 + n_1+\ldots +n_{j-1})},\ldots, \y^{(n_1+\ldots +n_{j})}]^T$ where $j=1,\ldots,p$ and $n_0=1$.}

For $j = 1,\ldots, p$, let us define a sensor model $\mathcal{Q}_j : L^2(\Omega,\mathbb{R}^{n_j}) \rightarrow L^2(\Omega,\mathbb{R}^{r_j})$ by the relation
\begin{equation}\label{eq3}
[\mathcal{Q}_j({\bf y}_j)](\omega)=Q_j[{\bf y}_j(\omega)]
\end{equation}
where $Q_j\in\mathbb{R}^{r_j\times n_j}$,
\begin{equation}\label{eq4}
r=r_1+\ldots+r_p,\;\text{ where }\;r_j\leq n_j.
\end{equation}
 Let us denote ${\bf u}_j=\mathcal{Q}_j({\bf y}_j)$ and ${\bf u}=[{\bf u}^{T}_1,\ldots,{\bf u}^{T}_p]^T$, where for $j=1,\ldots,p$, vector ${\bf u}_j\in L^2(\Omega,\mathbb{R}^{r_j})$ represents the compressed and filtered information vector transmitted by a $j$th sensor $\mathcal{Q}_j$ to the fusion center $\p$.
A fusion center model is defined by $\p : L^2(\Omega,\mathbb{R}^{r}) \rightarrow L^2(\Omega,\mathbb{R}^{m})$ so that
\begin{equation}\label{eq5}
[\p({\bf u})](\omega)=P[{\bf u}(\omega)],
\end{equation}
where $P\in\mathbb{R}^{m\times r}$ and ${\bf u}\in L^2(\Omega,\mathbb{R}^r)$. To state the problem in the next section, we also
denote
\begin{equation}\label{eq6}
  \e\left[\|{\bf x}(\omega)\|_2^2\right] := \|{\bf x}\|^2_\Omega =\int_\Omega \|{\bf x}(\omega)\|_2^2 d\mu(\omega) < \infty,
\end{equation}
where $\|{\bf x}(\omega)\|_2$ is the Euclidean norm of ${\bf x}(\omega)\in\mathbb{R}^m$. For convenience, we will use notation $\|{\bf x}\|^2_\Omega$, not $\e\left[\|{\bf x}(\omega)\|_2^2\right] $, to denote the norm in (\ref{eq6}).

\section{ Formalization and Statements of the  Problems } \label{preliminaries}

\subsection{Preliminaries and Formalization of the Problem} \label{discussion_problem}

For the WSN depicted in Fig. \ref{fig1}, the problem can be stated as follows:
Find models of the sensors, $\mathcal{Q}_1,\ldots,\mathcal{Q}_p$, and a model of the fusion center, $\p$, that provide
\begin{equation}\label{eq7}
\min_{\mathcal{P},\mathcal{Q}_1,\ldots,\mathcal{Q}_p} \left\| \x - \p\left[ \begin{array}{c}
                                                                                   \mathcal{Q}_1({\bf y}_1) \\
                                                                                   \vdots \\
                                                                                   \mathcal{Q}_p({\bf y}_p)
                                                                                 \end{array}\right] \right\|_\Omega^2
\end{equation}
under the assumption that $\mathcal{Q}_1,\ldots,\mathcal{Q}_p$ and $\p$ are given by (\ref{eq3}) and (\ref{eq5}), respectively.
The model of the fusion center, $\p$, can be represented as $\p = [\p_1,\ldots,\p_p]$ where  $\p_j : L^2(\Omega,\mathbb{R}^{r_j}) \rightarrow L^2(\Omega,\mathbb{R}^{m})$, for $j = 1,\ldots,p$. Let us write
\begin{eqnarray}
&&\hspace*{-20mm}\min_{\substack{\mathcal{P}_1,\ldots, \mathcal{P}_p, \\ \mathcal{Q}_1,\ldots,\mathcal{Q}_p}}\left\| \x - [\p_1,\ldots,\p_p]\left[ \begin{array}{c}
                               \mathcal{Q}_1({\bf y}_1) \\
                               \vdots \\
                               \mathcal{Q}_p({\bf y}_p)
                             \end{array}\right] \right\|_\Omega^2\nonumber \\
&& \hspace*{10mm}= \min_{\substack{\mathcal{P}_1,\ldots, \mathcal{P}_p, \\ \mathcal{Q}_1,\ldots,\mathcal{Q}_p}}\|{\bf x}-[\p_1\mathcal{Q}_1({\bf y}_1)+\ldots+\p_p\mathcal{Q}_p({\bf y}_p)]\|_\Omega^2\label{eq9}
\end{eqnarray}
where $\y = [ {\bf y}_1^T,\ldots,  {\bf y}_p^T ]^T.$


\subsection{Statement of the Problem} \label{stat1}

For $j=1,\ldots,p$, let us denote $\f_j = \p_j \q_j$. We also write $\f = [\f_1,\ldots,\f_p]$. Then the WSN model can be represented as
\begin{equation}\label{ftq}
\f (\y)= \p\q (\y) \qa\displaystyle \f(\y) = \sum_{j=1}^p \f_j(\y_j).
\end{equation}
Denote by $\rrr(m, n, k)$ the variety of all $m \times n$ linear operators of rank  at most $k$. For the sake of simplicity we sometimes will also write $\rrr_k$ instead of $\rrr(m, n, k)$.
The problem in (\ref{eq9}) can equivalently be reformulated as follows: Find $\f_1\in \mathcal{R}(m,n_1,r_1),$ $\ldots,$ $\f_p\in \mathcal{R}(m,n_p,r_p)$ that solve
\begin{equation}\label{f1p}
\min_{\mathcal{F}_1\in \mathcal{R}_{r_1},\ldots,\mathcal{F}_p\in\mathcal{R}_{r_p}} \left\|{\bf x}-\sum_{j=1}^p \f_j(\y_j)\right\|_\Omega^2.
\end{equation}
Recall that $r_j < n_j$, for $j=1,\ldots,p$.

Further, to simplify the notation we will use the same symbol to denote an operator
and the associated matrix. For example, we write $Q_j$ to denote both the operator $\mathcal{Q}_j$ and
matrix $Q_j$ introduced in (\ref{eq3}). Similarly, we write $P_j$ to denote both operator $\p_j $
 and matrix $P_j\in\mathbb{R}^{m\times r_j}$ introduced above, etc. In particular, by this reason, in
 (\ref{f1p}) we will write $F_j$ , not $\mathcal{F}_j$ .



\subsection{Assumptions}\label{assumptions}

For ${\bf x}$ represented by ${\bf x} = [{\bf x}^{(1)},\ldots, {\bf x}^{(m)}]^T$
 where $\x^{(j)}\in L^2(\Omega,\mathbb{R})$  we write
\begin{equation}\label{eq14}
E[{\bf xy}^T]=E_{xy}=\left\{\langle{\bf x}^{(j)},{\bf y}^{(k)}\rangle\right\}_{j,k=1}^{m,n}\in\mathbb{R}^{m\times n},
\end{equation}
where $\displaystyle\langle{\bf x}^{(j)}{\bf y}^{(k)}\rangle=\int_\Omega {\bf x}^{(j)}(\omega){\bf y}^{(k)}(\omega) d\mu(\omega)$ and ${\bf y} = [{\bf y}^{(1)},\ldots, {\bf y}^{(n)}]^T$.

The assumption used in the known methods \cite{dragotti2009, 5447742, 5504834,6334203,ma3098,4276987,Song20052131, 1420805, 4016296,4475387,Schizas2007,  Saghri2010} is that covariance matrices
$E_{xy}$ and $E_{yy}$  are known. At the same time, in many cases, it is
difficult to know exact values of $E_{xy}$ and $E_{yy}$.  For instance, if it is assumed that signal $\y$ is a Gaussian random vector then the associated  parameters $\rho$ and $\sigma^2$ for matrix $E_{yy}$ are still unknown (see \cite{4475387} as an example).  Therefore, we are mainly concerned with the case
when estimates of $E_{xy}$ and $E_{yy}$ can be obtained. Methods of estimation of matrices $E_{xy}$ and $E_{yy}$
were studied in a number of papers (see, for example, \cite{149980,Ledoit2004365,ledoit2012,Adamczak2009,Vershynin2012,won2013,schmeiser1991,yang1994}) and it is not a
subject of our work. In particular, samples of training signals taken for some different
random outcomes $\omega$ might be available. Some knowledge of the covariances can also
come from specific data models. Then the aforementioned covariance matrices can be
estimated.

In the following Section \ref{linear wsn}, we  provide solutions in terms of both matrices $E_{xy}$, $E_{yy}$ and their
estimates. The associated error analysis is given in  Sections \ref{error_analysis} and \ref{appendix}.

\subsection{Solution of Problem  (\ref{f1p}) for Single Sensor, i.e. for $p=1$}
\label{particular_solution}

Here, we recall the known result for the case of a single sensor, i.e. for $p = 1$ in (\ref{eq7}) and
(\ref{f1p}),  and provide some related explanations which will be extended in the sections that follow.

The Moore-Penrose generalized inverse for a matrix $M$ is denoted by $M^\dagger$. We set
$M^{1/2\dagger}=(M^{1/2})^\dagger$, where $M^{1/2}$ is a square root of $M$, i.e. $M = M^{1/2}M^{1/2}$.
Then the known solution (see, for example \cite{torbook2007}) of the particular case of problem (\ref{f1p}), for $p = 1$,  is given by
\begin{equation}\label{eq15}
F_1=\left[E_{x y_1}(E^\dagger_{y_1y_1})^{1/2}\right]_{r_1}(E^\dagger_{y_1y_1})^{1/2}+M_1[I-E_{y_1y_1}^{1/2}(E^\dagger_{y_1y_1})^{1/2}]
\end{equation}
where  symbol $[\cdot]_{r_1}$ denotes a truncated singular value decomposition (SVD) taken with $r_1$ first nonzero singular values
and matrix $M_1$ is  arbitrary. The expression (\ref{eq15}) represents the Karhunen-Lo\`{e}ve transform  as  given, for example, in \cite{dragotti2009}.

In (\ref{eq15}), in particular,  $M_1 = \mathbb{O}$  where $\mathbb{O}$ is the zero matrix. Then $Q_1$ and $P_1$ that solve (\ref{eq7})
follow from a decomposition of matrix $\left[E_{x,y_1}(E^\dagger_{y_1y_1})^{1/2}\right]_{r_1}(E^\dagger_{y_1y_1})^{1/2}$, based on the truncated
SVD, in a product $P_1Q_1$ of $r_1 \times n$ matrix $Q_1$ and $m \times r_1$ matrix $P_1$ , respectively.

We will extend this argument in Section \ref{linear wsn} below for a general case with more than one sensor.

\section{Main Results}\label{linear wsn}

\subsection{Greedy Approach  to Solution of Problem (\ref{f1p}), for $p=1,2,\ldots$}\label{greedy}

We wish to find $F_1,\ldots,F_p$ that provide a solution to the problem  (\ref{f1p}) for an arbitrary finite number of sensors in the WSN model, i.e. for $p=1,2,\ldots$ in (\ref{f1p}). To this end, we need  some more preliminaries which are given in Sections \ref{svdsvd} and \ref{reduction} that follow. The method itself and an associated algorithm are then represented in  Section \ref{structure}.

\subsubsection{SVD, Orthogonal Projections and Matrix Approximation}\label{svdsvd}   Let the SVD of a matrix $C \in \mathbb{R}^{m\times s}$ be given by
\begin{equation}\label{eq19}
C=U_C\Sigma_C V^T_C,
\end{equation}
where $U_C\in\mathbb{R}^{m\times m}$ and $V_C\in\mathbb{R}^{s\times s}$ are unitary matrices, $\Sigma_C=\text{diag}(\sigma_1(C),\ldots,\sigma_{\min(m,s)}(C))\in\mathbb{R}^{m\times s}$ is a generalized diagonal matrix, with the singular values $\sigma_1(C)\geq\sigma_2(C)\geq\ldots0$ on the main diagonal.
Let $U_C=[u_1\;u_2\;\ldots\;u_m]$ and $V_C=[v_1\;v_2\;\ldots\;v_s]$ be the representations of $U$ and $V$ in terms of their $m$ and $s$ columns, respectively. Let
\begin{equation}\label{lcrc}
L_C=\sum_{i=1}^{\text{rank }C}u_iu_i^T\in\mathbb{R}^{m\times m}\;\;\text{ and }\;\;R_C=\sum_{i=1}^{\text{rank }C}v_iv_i^T\in\mathbb{R}^{s\times s}
\end{equation}
 be the orthogonal projections on the range of $C$ and $C^T$, correspondingly. Define
\begin{equation}\label{eq21}
C_r=[C]_r=\sum_{i=1}^r\sigma_i(C)u_iv_i^T=U_{C_r}\Sigma_{C_r}V_{C_r}^T\in\mathbb{R}^{m\times s}
\end{equation}
for $r=1,\ldots,\text{rank }C$, where
\begin{equation}\label{eq22}
U_{C_r}=[u_1\;u_2\;\ldots\;u_r],\; \Sigma_{C_r}=\text{diag}(\sigma_1(C),\ldots,\sigma_r(C)) \text{ and } V_{C_r}=[v_1\;v_2\;\ldots\;v_r].
\end{equation}

For $r>\text{rank }C$, we write $C^{(r)}=C(=C_{\text{rank }C})$. For $1\leq r<\text{rank }C$, the matrix $C^{(r)}$ is uniquely defined if and only if $\sigma_r(C)>\sigma_{r+1}(C)$.

Consider the problem: Given matrices $S_j$ and $G_j$, for $j=1,\ldots,p$, find matrix $F_j$ that solves
\begin{eqnarray}\label{minfj}
 \min_{ F_j\in {\mathbb R}_{r_j}} \left\|S_j - F_jG_j \right\|^2, \quad \mbox{for $j=1,\ldots,p$}.
\end{eqnarray}
The solution is given by Theorem \ref{th11} (which is a particular case of the results in \cite{Torokhti2007,tor5277}) as follows.

\begin{theorem}\label{th11}
Let  $K_j= M_j\left(I-L_{G_{j}}\right)$ where  $M_j$ is an arbitrary matrix. The matrix $F_j$ given by
\begin{equation}\label{fjhjrgj}
F_j=\left[S_jR_{G_j}\right]_{r_j}G_j^\dagger (I+K_j), \quad \mbox{for $j=1,\ldots,p$,}
\end{equation}
is a minimizing matrix for the minimal problem (\ref{minfj}). Here,  $R_{G_j}$ and  $[\cdot]_{r_j}$  are defined similarly to  (\ref{lcrc}) and  (\ref{eq21}), respectively. Any minimizing $F_j$ has the
above form if and only if either
\begin{equation}\label{rrjaj}
r_j\geq \text{rank }(S_jR_{G_j})
\end{equation}
or
\begin{equation}\label{rjajrcj}
1\leq r_j<\text{rank }(S_jR_{G_j})\;\;\text{ and }\sigma_{r_j}(S_jR_{G_j})>\sigma_{{r_j}+1}(S_jR_{G_j}),
\end{equation}
 where $\sigma_{r_j}(S_jR_{G_j})$ is a singular value in the SVD for matrix $S_jR_{G_j}$.
\end{theorem}

\begin{proof}
The proof follows from  \cite{Torokhti2007,tor5277}.
  $\hfill\blacksquare$
\end{proof}

We note that the arbitrary matrix $M_j$ implies non-uniqueness of  solution (\ref{fjhjrgj}) of
problem (\ref{minfj}).

\subsubsection{Reduction of Problem (\ref{f1p}) to Equivalent Form}\label{reduction}

We denote $F=[F_1,\ldots,F_p]$, where $F\in\mathbb{R}^{m\times n}$ and $F_j\in\mathbb{R}^{m\times n_j}$, for all $j=1,\ldots,p$, and write $\|\cdot\|$ for the Frobenius norm. Then
\begin{eqnarray}\label{xfy1p}
\left\|{\bf x}-\sum_{j=1}^p F_j(\y_j)\right\|_\Omega^2 &=& \|{\bf x}-F({\bf y})\|_{\Omega}^2
  =  \text{tr}\{E_{xx}-E_{xy}F^T-FE_{yx}+FE_{yy}F^T\}\nonumber\\
                            & = & \|E_{xx}^{1/2}\|^2-\|E_{xy}(E_{yy}^{1/2})^\dagger\|^2+\|E_{xy}(E_{yy}^{1/2})^\dagger - FE_{yy}^{1/2}\|^2.
\end{eqnarray}
Denote by $\rt(m, n, k)$ the variety of all $m \times n$ matrices of rank  at most $k$. For the sake of simplicity we also write $\rt_k = \rt(m, n, k)$.

In (\ref{xfy1p}), only the last term depend on $F_1,\ldots,F_p$.   Therefore, (\ref{f1p}) and (\ref{xfy1p})  imply
\begin{equation}\label{fcaj}
  \min_{\mathcal{F}_1\in \mathcal{R}_{r_1},\ldots,\mathcal{F}_p\in\mathcal{R}_{r_p}} \left\|{\bf x}-\sum_{j=1}^p \f_j(\y_j)\right\|_\Omega^2
  =  \min_{F_1\in \mathbb{R}_{r_1},\ldots,F_p\in  \mathbb{R}_{r_p}}  \|E_{xy}(E_{yy}^{1/2})^\dagger - FE_{yy}^{1/2}\|^2.
\end{equation}
Let us now denote $H=E_{xy}(E_{yy}^{1/2})^\dagger$ and represent matrix $E_{yy}^{1/2}$ in blocks,
$ 
E_{yy}^{1/2}= [G_{1}^T,\ldots,G_{p}^T ]^T
$ 
where $G_j\in \rt^{n_j\times n}$, for $j=1,\ldots,p$. Then in (\ref{fcaj}),
\begin{eqnarray}\label{pezzghj}
E_{xy}(E_{yy}^{1/2})^\dagger - FE_{yy}^{1/2} = H -\sum_{j=1}^p F_jG_j .
\end{eqnarray}
Therefore, (\ref{xfy1p}) and (\ref{pezzghj}) imply
\begin{eqnarray}\label{f1pf1p}
\min_{\mathcal{F}_1\in \mathcal{R}_{r_1},\ldots,\mathcal{F}_p\in\mathcal{R}_{r_p}} \left\|{\bf x}-\sum_{j=1}^p \f_j(\y_j)\right\|_\Omega^2 = \min_{F_1\in {\mathbb R}_{r_1},\ldots, F_p\in {\mathbb R}_{r_1}}\left\|H - \sum_{j=1}^p F_jG_j\right\|^2.
\end{eqnarray}
On the basis of (\ref{f1pf1p}), problem (\ref{f1p}) and the problem
\begin{eqnarray}\label{f1ph1p}
 \min_{F_1\in {\mathbb R}_{r_1},\ldots, F_p\in {\mathbb R}_{r_1}}\left\|H - \sum_{j=1}^p F_jG_j \right\|^2.
\end{eqnarray}
are equivalent. Therefore, below we consider problem  (\ref{f1ph1p}). In (\ref{f1ph1p}),  for $j=1,\ldots,p$, we use the representation
 \begin{eqnarray}\label{}
\left\|H - \sum_{j=1}^p F_jG_j \right\|^2 = \left\|S_j - F_jG_j\right\|^2,
\end{eqnarray}
where  $\displaystyle S_j = H-\sum_{\substack{i=1\\i\neq j}}^pF_iG_i$.

\subsubsection{Greedy Method for Solution of Problem (\ref{f1ph1p})}\label{structure}

A solution of problem (\ref{f1ph1p}) is based on the idea of the MBI method \cite{chen87,Zhening2015} which is a greedy approach to  solving optimization problems. The advantages of the MBI method  have been mentioned in Section \ref{differences}.
To begin with, let us denote
$$
\rt_{r_1,\ldots,r_p} = \rt_{r_1}\times...\times\rt_{r_p}, \quad{\bfF}=\left(F_1,...,F_p\right)\in\rt_{r_1,\ldots,r_p}\qa f({\bfF}) = \left\|H - \sum_{j=1}^p F_jG_j \right\|^2.
$$
The method we consider consists of the following steps.

{\em $1$st step.}  Given ${\bfF}^{(0)}=(F^{(0)}_1,...,F^{(0)}_p)\in\rt_{r_1,\ldots,r_p},$ compute, for $j=1,\ldots,p$,
 \begin{eqnarray}\label{whf1j}
\widehat{F}^{(1)}_j = \left[S^{(0)}_jR_{G_j}\right]_{r_j}G_j^\dagger (I+K_j),\quad \mbox{where $\displaystyle S^{(0)}_j= H-\sum_{\substack{i=1\\i\neq j}}^pF^{(0)}_iG_i$}.
\end{eqnarray}
Note that $\widehat{F}^{(1)}_j$ is the solution of problem (\ref{minfj}) represented by (\ref{fjhjrgj}). A choice of $F^{(0)}_1,\ldots, F^{(0)}_p,$ is considered in Section \ref{initial_iterations} below.

{\em $2$nd step.} Denote
  \begin{eqnarray}\label{bf1j}
  \bar{\bfF}^{(1)}_j = \left(F^{(0)}_1,\ldots,F^{(0)}_{j-1},\widehat{F}_j^{(1)},F^{(0)}_{j+1},\ldots,F^{(0)}_{p}\right),
\end{eqnarray}
 select $\bar{\bfF}^{(1)}_k$ such that
\begin{eqnarray}\label{bf1k}
  f(\bar{\bfF}^{(1)}_k) = \min_{\mbox{\scriptsize\boldmath $\bar{F}$}^{(1)}_1,\ldots, \mbox{\scriptsize\boldmath $\bar{F}$}^{(1)}_p} \left \{f(\bar{\bfF}^{(1)}_1), \ldots, f(\bar{\bfF}^{(1)}_k),\ldots, f(\bar{\bfF}^{(1)}_p) \right\}
\end{eqnarray}
and write
\begin{eqnarray}\label{bff1}
{\bfF}^{(1)} = \bar{\bfF}^{(1)}_k,
\end{eqnarray}
where we denote ${\bfF}^{(1)} = (F^{(1)}_1,...,F^{(1)}_p)\in\rt_{r_1,\ldots,r_p}$.

Then we repeat procedure (\ref{whf1j})-(\ref{bff1}) with the replacement of ${\bfF}^{(0)}$ by ${\bfF}^{(1)}$  as follows: Given ${\bfF}^{(1)} = (F^{(1)}_1,...,F^{(1)}_p),$ compute
$$
\widehat{F}^{(2)}_j = \left[S^{(1)}_jR_{G_j}\right]_{r_j}G_j^\dagger (I+K_j),
\quad \mbox{where $\displaystyle S^{(1)}_j= H-\sum_{\substack{i=1\\i\neq j}}^pF^{(1)}_iG_i$},
 $$
 select $\bar{\bfF}^{(2)}_k$ that satisfies (\ref{bf1k}) where superscript $(1)$ is replaced with  superscript $(2)$, and set  ${\bfF}^{(2)} = \bar{\bfF}^{(2)}_k.$

This process is continued up to the $q$th step when a given tolerance $\epsilon \geq 0$  is achieved in the sense
\begin{eqnarray}\label{qeps1}
|f({\bfF}^{(q+1)}) - f({\bfF}^{(q)}) |\leq \epsilon, \quad \mbox{for $q =1,2,\ldots$}.
\end{eqnarray}
It is summarized as follows.

\vspace{5mm}

\hrule

\vspace{0.1cm}

\noindent{\bf Algorithm 1:} Greedy solution of problem  (\ref{f1ph1p})

\vspace{0.1cm}

\hrule

\vspace{0.1cm}

\noindent{\bf Initialization:} ${\bfF}^{(0)}$,
$H$,  $S_1,...,S_p$ and $\epsilon>0$.

\vspace{0.25cm}

\hrule

\vspace{0.25cm}

\noindent\hspace{0.1cm} {1.} \hspace{0.1cm} {\bf for }$q=0,1,2,...$

\vspace{0.05cm}

\noindent\hspace{0.1cm} {2.} \hspace{0.5cm} {\bf for }$j=1,2,...,p$

\vspace{0.05cm}

\noindent\hspace{0.1cm} {3.} \hspace{0.9cm} $S^{(q)}_j= H-\displaystyle \sum_{\substack{i=1\\i\neq j}}^pF^{(q)}_iG_i$

\vspace{0.05cm}

\noindent\hspace{0.1cm} {4.} \hspace{0.9cm}  $\widehat{F}^{(q+1)}_j = \left[S^{(q)}_jR_{G_j}\right]_{r_j}G_j^\dagger$

\vspace{0.05cm}

\noindent\hspace{0.1cm} {5.} \hspace{0.9cm}  $\bar{\bfF}^{(q+1)}_j = \left(F^{(q)}_1,...,F^{(q)}_{j-1},\widehat{F}_j^{(q+1)},F^{(q)}_{j+1},...,F^{(q)}_{p}\right)$

\vspace{0.05cm}

\noindent\hspace{0.1cm} {6.} \hspace{0.5cm} {\bf end}

\vspace{0.05cm}

\noindent\hspace{0.1cm} {7.} \hspace{0.1cm} Choose ${\bfF}^{(q+1)} = \bar{\bfF}^{(q+1)}_k$ where  ${\bfF}^{(q+1)}= (F^{(q+1)}_1,...,F^{(q+1)}_p)$ and $\bar{\bfF}^{(q+1)}_k$ is such that\\
 \hspace*{25mm}$ \displaystyle f(\bar{\bfF}^{(q+1)}_k) = \min_{\mbox{\scriptsize\boldmath $\bar{F}$}^{(q+1)}_1,\ldots, \mbox{\scriptsize\boldmath $\bar{F}$}^{(q+1)}_p} \left\{f(\bar{\bfF}^{(q+1)}_1), \ldots, f(\bar{\bfF}^{(q+1)}_k),\ldots,  f(\bar{\bfF}^{(q+1)}_p)\right\}.$

\vspace{0.05cm}

\noindent\hspace{0.1cm} {7.} \hspace{0.1cm} {\bf If} $|f({\bfF}^{(q+1)}) - f({\bfF}^{(q)}) |\leq \epsilon$

\vspace{0.05cm}

\noindent\hspace{0.1cm} {8.} \hspace{0.5cm} {\bf Stop}

\vspace{0.05cm}

\noindent\hspace{0.1cm} {9.} \hspace{0.1cm} {\bf end}

\vspace{0.05cm}

\noindent\hspace{0.1cm} {10.} {\bf end}

\vspace{0.20cm}

\hrule

\vspace{1cm}



Algorithm 1 converges to a coordinate-wise minimum point of objective function $f({\bfF}) = \left\|H - \displaystyle\sum_{j=1}^p F_jG_j \right\|^2$ which is its local minimum.  Section \ref{convergence} provides more associated  details.
 \vspace{2mm}

\begin{remark}\label{rem1}
For $p=1$, the formulas for $F_j$ in (\ref{fjhjrgj}) and (\ref{whf1j}) coincide, and take the form
\begin{equation}\label{f1h1g1}
F_1 = \left[HR_{G_{1}}\right]_{r_1}G_{1}^\dagger(I+K_1)
\end{equation}
where  $G_{1} = E_{y_1y_1}^{1/2}$, $K_1=M_1\left(I-L_{G_{1}}\right)$, $L_{G_{1}} = E_{y_1y_1}^{1/2}(E^\dagger_{y_1y_1})^{1/2}$ and $HR_{G_{1}} = E_{x y_1}(E^\dagger_{y_1y_1})^{1/2}$. Thus, $F_1$ in (\ref{f1h1g1}) coincides with $F_1$ in (\ref{eq15}) (since $M_1$ is arbitrary). In other words, the KLT represented by  (\ref{eq15}) is a particular case of the expressions in (\ref{fjhjrgj}) and (\ref{whf1j}). For this reason, (\ref{fjhjrgj}) and (\ref{whf1j})  can be regarded as  extensions of the KLT to the case under consideration. Therefore, the transform represented by  $\displaystyle \sum_{j=1}^p F_j(\y_j)$ where $F_1,\ldots,F_p$ solve (\ref{minfj}) can be regarded { the multi-compressor  Karhunen-Lo\`{e}ve-like transform} (or the multi-compressor  KLT). Further, Algorithm 1 represents the version of the MBI method that uses the multi-compressor  KLT. Therefore, the WSN model in the form  $\displaystyle \sum_{j=1}^p F_j^{(q+1)}(\y_j)$ where $F^{(q+1)}_1,\ldots,F^{(q+1)}_p$ are determined by Algorithm 1 can be interpreted as a transform as well. We call this transform   the multi-compressor  KLT-MBI.
\end{remark}

\begin{remark}\label{rem2}
In practice, an exact representation of matrices $E_{xy}$ and $E_{yy}$ might be unknown. Their
estimates, $\widetilde{E}_{xy}$ and $\widetilde{E}_{yy}$,  can be obtained by known methods \cite{149980,Ledoit2004365,ledoit2012,Adamczak2009,Vershynin2012,won2013,schmeiser1991,yang1994}.
In this regard, we denote $\widetilde{H}= \widetilde{E}_{xy}(\widetilde{E}_{yy}^{1/2})^\dagger$ and $\widetilde{E}_{yy}^{1/2}= [\widetilde{G}_{1}^T,\ldots,\widetilde{G}_{p}^T ]^T$
where $\widetilde{G}_j\in \rt^{n_j\times n}$, for $j=1,\ldots,p$, is a block of $\widetilde{E}_{yy}^{1/2}$. Then in (\ref{whf1j}),  $G_j$, $H$ and $G_i$ should be replaced with  $\widetilde{G}_j$, $\widetilde{H}$ and $\widetilde{G}_i$, respectively. In this case,  steps (\ref{whf1j})-(\ref{qeps1}) of the method and  Algorithm 1 are the same as before but ${\bfF}^{(q+1)}= (F^{(q+1)}_1,...,F^{(q+1)}_p)$ should be denoted by $\widetilde{\bfF}^{(q+1)}= (\widetilde{F}^{(q+1)}_1,...,\widetilde{F}^{(q+1)}_p).$
\end{remark}



\subsection{Models of Sensors and Fusion Center}\label{linear models}

The models of sensors $Q_1,\ldots,Q_p$ and the fusion center $P=[P_1,\ldots,P_p]$  of the WSN in Fig. \ref{fig1} follow from the formula $F_j = P_j Q_j$ introduced in (\ref{ftq}).
By the proposed method, $F_j$  is represented by $F^{(q+1)}_j$ given by  Algorithm 1 above. Therefore,  the mathematical model of the   WSN is given by
\begin{equation}\label{eq39}
\sum_{j=1}^p {F}^{(q+1)}_j({\bf y}_j)=\sum_{j=1}^p {P}^{(q+1)}_j {Q}^{(q+1)}_j({\bf y}_j)={P}^{(q+1)}\left(
                                                                          \begin{array}{c}
                                                                            {Q}_1^{(q+1)}({\bf y}_1) \\
                                                                            \vdots \\
                                                                            {Q}_p^{(q+1)}({\bf y}_p) \\
                                                                          \end{array}
                                                                        \right).
\end{equation}
Here, ${F}^{(q+1)}_j={P}_j^{(q+1)}{Q}_j^{(q+1)}$ and ${P}^{(q+1)}=[{P}_1^{(q+1)},\ldots,{P}_p^{(q+1)}]$.
Then ${Q}^{(q+1)}_j\in \rt^{r_j \times n_j}$ and ${P}^{(q+1)}_j\in \rt^{m \times r_j}$ follow from the  SVD decomposition of matrix ${F}^{(q+1)}_j$ taken with $r_j$ first nonzero singular values.  This procedure has been described in Section \ref{particular_solution} for the
case of only one sensor, i.e. for $p = 1$.

For the case when instead of matrices $E_{xy}$ and $E_{yy}$ their estimates $\widetilde{E}_{xy}$ and $\widetilde{E}_{yy}$ are used, the models are constructed by the similar procedure with the replacement of $F^{(q+1)}_j$ by $\widetilde{F}^{(q+1)}_j$ (see Remark \ref{rem2} above).

Note that the proposed WSN model  also provides de-noising of observations ${\bf y}_1,\ldots, {\bf y}_p$.

\section{Determination of Initial Iterations }
\label{initial_iterations}

To start Algorithm 1, the values of initial iterations $F_j^{(0)}$ and $\widetilde{F}_j^{(0)}$ should be defined. It is done as follows.
Let us denote $\x = [\x^T_1,\ldots,\x_p^T]^T$ where  $\x_i \in L^{2}(\Omega,{\mathbb R}^{m_i})$,  $i=1,\ldots,p$ and $m_1 + \ldots +m_p = m$. Suppose that matrix $P\in \rt^{m\times r}$ is given by
$
 P=\diag (P_{11},\ldots, P_{pp})
$
where $P_{jj}\in\rt^{m_j\times r_j}$, for $j=1,\ldots,p$.  Then
\begin{eqnarray}\label{b}
&&\hspace{-30mm}\left\|\left[ \begin{array}{c}
\x_1 \\
\vdots\\
\x_p\end{array} \right] - \diag (\p_{11},\ldots, \p_{pp}) \left[ \begin{array}{c}
\q_1 (\y_1) \\
\vdots\\
\q_p(\y_p)\end{array} \right] \right\|^2_{\Omega}\nonumber \\
&& \hspace{20mm} = \left\|\left[ \begin{array}{c}
\x_1 - \f_1(\y_1)\\
\vdots\\
\x_p - \f_p(\y_p)\end{array} \right]  \right\|^2_{\Omega}\label{bxp}\label{nxpfpyp}
= \sum_{j=1}^p \|\x_j - \f_j (\y_j)\|^2_{\Omega}\label{sumxjfj}
 \end{eqnarray}
and
\begin{eqnarray}\label{minf1pp}
  \min_{\mathcal{F}_1\in \mathcal{R}_{r_1},\ldots, \mathcal{F}_p\in  \mathcal{R}_{r_p}}\sum_{j=1}^p \|\x_j - F_j (\y_j)\|^2_{\Omega}
  = \sum_{j=1}^p \min_{F_j\in \mathcal{R}_{r_j}} \|\x_j - \f_j (\y_j)\|^2_{\Omega}.
 \end{eqnarray}
As a result, in this case, problem (\ref{f1p}) is reduced to the  problem of finding $\f_j$ that solves
 \begin{eqnarray}\label{f1pp1}
 \min_{\mathcal{F}_j\in \mathcal{R}_{r_j}} \left\|\x_j -  \mathcal{F}_j (\y_j)\right\|^2_{\Omega},
  \end{eqnarray}
for $j=1,\ldots,p$. Its solution is given by (\ref{eq15}) in Section \ref{particular_solution} above, i.e. by
\begin{equation}\label{fj02}
\widehat{F}_j=\left[E_{x_j y_j}(E^\dagger_{y_jy_j})^{1/2}\right]_{r_j}(E^\dagger_{y_jy_j})^{1/2} + M_j[I-E_{y_jy_j}^{1/2}(E^\dagger_{y_jy_j})^{1/2}]
\end{equation}
where $M_j$ is an arbitrary matrix. Then the initial iterations for Algorithm 1 are defined by
 \begin{eqnarray}\label{fj0xyj}
F_j^{(0)} = \widehat{F}_j,\quad \text{for $j=1,\ldots,p$}.
\end{eqnarray}
Similarly, when instead of matrices $E_{xy}$ and $E_{yy}$ their estimates $\widetilde{E}_{xy}=\{\widetilde{E}_{x_i y_j}\}_{i,j=1}^p$ and $\widetilde{E}_{yy}=\{\widetilde{E}_{y_i y_j}\}_{i,j=1}^p$ are used, the initial iterations are defined by
 \begin{eqnarray}\label{tfj0xyj}
\widetilde{F}^{(0)}_j = \left[\widetilde{E}_{x_j y_j}(\widetilde{E}^\dagger_{y_jy_j})^{1/2}\right]_{r_j}(\widetilde{E}^\dagger_{y_jy_j})^{1/2}  + M_j[I-\widetilde{E}_{y_jy_j}^{1/2}(\widetilde{E}^\dagger_{y_jy_j})^{1/2}],
\end{eqnarray}
where $\widetilde{E}_{x_j y_j}$ and $\widetilde{E}_{y_jy_j}$ are blocks of $\widetilde{E}_{xy}$ and $\widetilde{E}_{yy}$, respectively.

\section{Error Analysis: A Posteriori Associated Errors}
\label{error_analysis}

For  $F^{(q+1)}=[F_1^{(q+1)},\ldots,F_p^{(q+1)}]$ determined by Algorithm 1, the error associated
 with the proposed WSN  model   is represented as
\begin{eqnarray}\label{eq71}
 \|{\bf x}-[F_1^{(q+1)},\ldots,F_p^{(q+1)}]({\bf y})\|^2_\Omega = \|E_{xx}^{1/2}\|^2-\|E_{xy}(E_{yy}^{1/2})^\dagger\|^2
                                                           +\|E_{xy}(E_{yy}^{1/2})^\dagger-F^{(q+1)}E_{yy}^{1/2}\|^2,\nonumber
\end{eqnarray}
where $F^{(q+1)}=[F_1^{(q+1)},\ldots,F_p^{(q+1)}]$.

For $\widetilde{F}_j^{(q+1)}$ described in Remark \ref{rem2}, the associated error is given by the similar expression:
\begin{equation}\label{eq72}
\|{\bf x}-[\widetilde{F}_1^{(q+1)},\ldots,\widetilde{F}_p^{(q+1)}]({\bf y})\|^2_\Omega  =  \|\widetilde{E}_{xx}^{1/2}\|^2
-\|\widetilde{E}_{xy}(\widetilde{E}_{yy}^{1/2})^\dagger\|^2+\|\widetilde{E}_{xy}(\widetilde{E}_{yy}^{1/2})^\dagger
-\widetilde{F}^{(q+1)}\widetilde{E}_{yy}^{1/2}\|^2,
\end{equation}
where $\widetilde{F}^{(q+1)}=[\widetilde{F}_1^{(q+1)},\ldots,\widetilde{F}_p^{(q+1)}]$. The above formula (\ref{eq72}) is used in our simulations represented in the following section.


\section{Simulations}\label{simulations}

Here, we wish to illustrate the advantages of the proposed methodology with numerical examples carried out
{ under the assumption that either covariance matrices  $E_{xy}$, $E_{yy}$ or their estimates are known.  The assumption that only the covariance matrices are known is
similar to that used in \cite{dragotti2009, 5447742, 5504834,6334203,ma3098,4276987,Song20052131, 1420805, 4016296,4475387,Schizas2007,  Saghri2010}.  In particular, the estimates can be obtained  from samples of training signals.}
In many situations, the number of samples, $s$, is often
smaller  than the dimensions of the signals ${\bf x}$ and ${\bf y}$, which are $m$ and $n$, respectively \cite{won2013}.
At the same time, it is known that as $s\rightarrow \infty$, the ergodic theorem asserts that the  estimates converge to the true matrix values   \cite{schmeiser1991,yang1994}. In particular, for large $s$,  the  estimates of the covariance matrix  have been considered in \cite{Ledoit2004365,ledoit2012}. It is interesting to compare our simulation results for the cases when  $s$ is `relatively' small and   `relatively' large.  In the examples that follow, both case are considered.

{ A comparison with  known methods  \cite{ma3098,4276987,Song20052131, 1420805, 4016296,4475387} is a s follows. The method \cite{4475387} represents a generalization of methods \cite{Song20052131, 1420805, 4016296} and therefore, we provide a numerical comparison with method \cite{4475387} which includes, in fact, a comparison with  methods  \cite{Song20052131, 1420805, 4016296} as well. Further, covariance matrices used in the  simulations associated with  Figs. \ref{ex:fig1} (a), (c) and Figs. \ref{ex:fig2} (b), (c), (d)   are singular, and therefore, method \cite{4276987} is not applicable (in this regard, see also Section \ref{assumptions}). Therefore, in Figs. \ref{ex:fig1} (a), (c) and Figs. \ref{ex:fig2} (b), (c), (d), results related to algorithm in \cite{4276987} are not given.
 By the same reason, the method presented in \cite{ma3098} is not applicable as well. Moreover, the method in \cite{ma3098} is restricted to the case when the covariance matrix formed by the noise vector is block diagonal which is not the case here.
}

In the examples below,  different types of noisy observed signals  and different compression ratios are considered. In all  examples, our method provides the better associated accuracy than that for the methods in \cite{4276987,4475387} (and methods in \cite{Song20052131, 1420805, 4016296} as well, because they follow from  \cite{4475387}).

\begin{example}\label{ex0}
  We start with an example  similar to that considered in \cite{4475387}  assuming that a WSN has two sensors and the observations  $\y_1$ and $\y_2$ are represented by
\begin{equation}
{\bf y}_1=\x+\xi_1,\;\;\text{ and }\;\;{\bf y}_2=\x+\xi_2,
\end{equation}
where ${\bf x}\in L^2(\Omega,\mathbb{R}^3)$,  $\xi_1\in L^2(\Omega,\mathbb{R}^3)$ and $\xi_2\in L^2(\Omega,\mathbb{R}^3)$ are Gaussian {\em independent} random vectors with the zero mean. Let
$E_{xx}=\left[
  \begin{array}{ccc}
    0.585 & 0.270 & 0.390\\
    0.270 & 0.405 & 0.180\\
    0.390 & 0.180 & 0.260\\
  \end{array}
\right]$ and $E_{\xi_j,\xi_j} = \sigma_j^2 I_3$, for $j=1,2$, where $\sigma_1=0.2$ and $\sigma_2=0.4$, and $I_3$ is the $3\times 3$ identity matrix. Then $E_{xy}=[E_{xx}\;E_{xx}]$
and
$E_{yy}=\left[
         \begin{array}{cc}
           E_{xx}+\sigma_1^2 I_3 & E_{xx} \\
           E_{xx} & E_{xx} +\sigma_2^2 I_3\\
         \end{array}
       \right].$
For $r_1=r_2=1$, Algorithm 1 requires three iterations to achieve tolerance $\epsilon =0.146.$
The achievable  tolerance  of methods \cite{4276987,4475387}, for $r_1=r_2=1$, is $18\%$ worse, $\epsilon =0.173$, and it  is not improved after the initial iteration proposed in \cite{4475387}. 

\end{example}

\begin{example}\label{ex1}
{\rm  Let us consider the case  of a WSN with two sensors again where, as before,
\begin{equation}
{\bf y}_1=\x+\xi_1,\;\;\text{ and }\;\;{\bf y}_2=\x+\xi_2,
\end{equation}
where ${\bf x}\in L^2(\Omega,\mathbb{R}^m)$, $\xi_1\in L^2(\Omega,\mathbb{R}^m)$ and $\xi_2\in L^2(\Omega,\mathbb{R}^m)$, and $\y_1$ and $\y_2$ are noisy versions of the source. { Unlike Example \ref{ex0}  we now  assume that covariance matrices ${E}_{xy}$ and ${E}_{yy}$ are unknown. Therefore, their estimates $\widetilde{E}_{xy}$ and $\widetilde{E}_{yy} = \left\{\widetilde{E}_{y_iy_j}\right\}$, $i,j=1,2$, should be used. To this end,
}   estimates $\widetilde{E}_{xy}$ and $\widetilde{E}_{yy}$ have been determined from the samples of training signals as follows:
\begin{eqnarray}\label{ex1exy}
\widetilde{E}_{xy} = \frac{1}{s} [XY_1^T,  XY_2^T ]  \qa \widetilde{E}_{y_iy_j} = \frac{1}{s}Y_i Y_j^T\quad \mbox{for $i,j=1,2$}.
\end{eqnarray}
Here, $X\in\rt^{m\times s}$ has uniformly distributed random entries and, for $i=1,2$,
 \begin{eqnarray}\label{ex1:yi}
 Y_i = X + \sigma_i {\up}_i,
 \end{eqnarray}
  where $\sigma_i\in\rt$ and ${\up}_i\in\rt^{m\times s} $ has random entries, chosen from a normal distribution with mean
zero and variance one. Diagrams of typical errors associated with the proposed Algorithm 1 and known methods \cite{4276987,4475387} are given in Figs. \ref{ex:fig1} (a), (b).

Note that in Figs. \ref{ex:fig1} (a), (b), the obtained results are illustrated for different compression ratios $c_j = r_j/n_j$ where $j=1,\ldots,p$. The compression ratios $c_j=1/5$, for $j=1,2,3$,   used to obtain the results represented in  Fig. \ref{ex:fig1} (b) are smaller than those in Fig. \ref{ex:fig1} (a), $c_1=3/5$ and $c_2=7/10$. This is a reason for the error magnitudes represented in Fig. \ref{ex:fig1} (a) being  smaller than those in Fig. \ref{ex:fig1} (b). This observation also holds for other examples that follow.
Further, in Fig. \ref{ex:fig1} (b),  due to large sample size, $s=10000$, the estimate of matrix $E_{y_iy_j}$ is very close to its true value which is the identity. By this reason, iterations of our method are similar to each other and the associated errors are similar for almost all iterations. The same effect holds for methods \cite{4276987,4475387}.}
\end{example}

\begin{example}\label{ex2}
{\rm
Here, we consider the case when observations $\y_1$ and $\y_2$ are very noisy, i.e. reference signal $\x$ is significantly suppressed. To this end, we do not assume that $Y_i\in\rt^{m\times s}$ is  represented in the form (\ref{ex1:yi}) but it has random entries, chosen from a normal distribution with mean zero and variance one.
We also use  $\widetilde{E}_{xy}$ and $\widetilde{E}_{yy}= \left\{\widetilde{E}_{y_iy_j}\right\}$, $i,j=1,2$ in the form (\ref{ex1exy}) as before where $X\in\rt^{m\times s}$ is as in the above Example \ref{ex1}. For $m=n_1=N_2=2$ and $s=4$, examples of those matrices are
$
X = \left[
  \begin{array}{cccc}
   0.086  &  0.439  &  0.857  &  0.904\\
   0.074  &  0.574  &  0.386  &  0.429
  \end{array}
\right],
$
$
Y_1=\hspace*{-1mm}\left[
  \begin{array}{rrrr}
0.284  & -0.942  &   0.067  &  0.222\\
-2.206 &   0.514 &  -1.293 &  -0.686
  \end{array}
\right]$ and
$Y_2= \hspace*{-1mm}\left[
  \begin{array}{cccc}
0.4660  & -0.1260  &  0.3870  &  0.3290\\
0.6880  & -0.4690  & -0.9420  & -0.5630
  \end{array}
\hspace*{-1mm}\right].
$

For the case of two sensors (i.e. for $p=2$) and for the above samples $X$, $Y_1$ and $Y_2$,  the errors associated with the proposed method and known method \cite{4475387} are given in Fig. \ref{ex:fig1} (c).

For larger  magnitudes of $m,n_1,n_2,n_3,s$ and $r_j$, for $j=1,2,3$,  the errors associated with the proposed method and known methods are represented, for the case of two and three sensors (i.e. for $p=2$ and $p=3$, respectively), in Fig. \ref{ex:fig1} (d) and Fig. \ref{ex:fig2} (a).}

\end{example}

\begin{example}\label{ex3}
{\rm
In this example, we consider the case when  observations are corrupted by noise in the  way which is different from those in Examples \ref{ex1} and \ref{ex2}. Namely, we assume that, for $j=1,\ldots,p$,
\begin{equation}\label{}
\y_j = \aaa_j \x + \xi_j
\end{equation}
where $\aaa_j: L^2(\Omega,\mathbb{R}^{m}) \rightarrow L^2(\Omega,\mathbb{R}^{m})$ is a linear operator defined by  matrix $A_j\in \rt^{m\times m}$ with uniformly distributed random entries, and $\xi_j$ is a random noise. Samples of $\x$ and $\xi_j$ are simulated as matrices $X\in\rt^{m\times s}$ and $\sigma_j\up_j\in\rt^{m\times s}$, respectively, where $\sigma_j\in\rt$, such that $X$ has  uniformly distributed random entries and $\up_j$  has random entries, chosen from a normal distribution with mean zero and variance one.

The errors associated with  the proposed method and the known method, for the case of two and three sensors (i.e. for $p=2$ and $p=3$, respectively), and different choices of $m,n_j,s$ and $r_j$, for $j=1,2$ and $j=1,2,3$, are represented in Figs. \ref{ex:fig2} (b), (c) and (d).

Note that method \cite{4475387} is not numerically stable in these simulations. We believe this is because of the reason mentioned in Section \ref{differences}.}

\end{example}

\begin{example}\label{ex4}
{\rm
In the above examples, we used estimates of training signals, not training signals themselves.
Here, we wish to illustrate the obtained  theoretical results in a different way, by a comparison of a training reference signal  with its estimates obtained by our method and known methods.
To this end, we simulate the training reference signal $\x$  by its realizations, i.e. by a matrix $\bfX\in\rt^{m\times k} $ where each column  represents a realization of the signal. A sample $X\in\rt^{m\times s} $ with $s < k$ is formed from $\bfX$ by choosing the even columns. To represent the obtained results in a visible way, signal $\bfX\in\rt^{m\times k}$ is chosen as the known image Lena given by the $128\times 128$ matrix -- see Fig. \ref{ex4:fig1} (a), i.e. with $m, k =128$. Then $X\in\rt^{128\times 64} $.

Further, we consider the WSN with two sensors, i.e. with $p=2$, where the observed signal $\bfY_j$, for $j=1,2$, is simulated as follows:
\begin{eqnarray*}
\bfY_j = A_j\hspace*{-1mm}*\hspace*{-1mm}\bfX + \sigma_j \up_j
\end{eqnarray*}
where $A_j\in\rt^{128\times 128} $ has  uniformly distributed random entries, $\up_j\in\rt^{128\times 128} $  has random entries, chosen from a normal distribution with mean zero and variance one, $A_j\hspace*{-1mm}*\hspace*{-1mm} X $ represents the Hadamard matrix product, and $\sigma_1=0.2$ and $\sigma_2=0.1$.
Estimates $\widetilde{E}_{xy}$ and $\widetilde{E}_{y_iy_j}$ are used in the form (\ref{ex1exy}) where sample $Y_j\in \rt^{128\times 64}$ is formed from $\bfY_j$ by choosing the even columns.

 For $r_1= r_2=64$, the simulation results  are represented in Figs. \ref{ex4:fig1} and \ref{ex4:fig2}. Our method and known methods in \cite{4475387} and \cite{4276987}) have been applied to the above signals with $50$ iterations each. The associated errors  are evaluated in the form $|\bfX - \widehat{\bfX}|$ where $\widehat{\bfX}$ is the reconstruction of $\bfX$ by  the  method we use (i.e. by our method or methods in \cite{4276987} and \cite{4475387}).

 Similar to the other examples, Figs. \ref{ex4:fig1} and \ref{ex4:fig2} demonstrate  a more accurate   signal reconstruction associated  with the proposed method  than that associated with known methods.

}

\end{example}

\section{Conclusion}\label{conclusion}

We have addressed the problem of estimating an unknown random vector source when
the vector cannot be observed centrally. In this scenario, typical of wireless sensor networks (WSNs), distributed sensors are aimed to filter
and compress noisy observed vector, and then the compressed signals are transmitted to
the fusion center that decompress the signals in such a way that the original vector is
estimated within a prescribed accuracy. {  The key problem is to find models of the sensors and the fusion center in the best possible way}.

We proposed and justified the method for the determination of the models { based on a combination of  the solution \cite{tor5277} of the  rank constrained
least squares minimizing problem (represented by (\ref{f1p}) in Section \ref{stat1}) and  the  maximum block improvement  (MBI) method  \cite{chen87,Zhening2015}. The proposed
method is based on the following steps. First, we have shown how the original problem
can be reduced to the form (\ref{fcaj}) (Section \ref{reduction}) that allowed us to use the approaches developed in \cite{chen87,Zhening2015,tor5277}. As a result, under the assumption that  the associated covariance matrices or their estimates are known (from testing experiments, for example), the procedure for determining models of the sensors and the fusion center is given by Algorithm 1 (Section \ref{structure}).

 The obtained optimal WSN model represents an extension of the Karhunen-Lo\`{e}ve
transform (KLT) and has been  called the  multi-compressor KLT-MBI.
The known KLT follows from the multi-compressor  KLT-MBI as a particular case.}
  The models of the sensors and the fusion  center  have been determined in terms of the pseudo-inverse matrices.
Therefore, the proposed models are always well determined and numerically stable.  In other words, the proposed WSN  models provide compression, de-noising and reconstruction of distributed signals for the cases  when known  methods either are not applicable   or produce  larger associated errors. As a result, this approach mitigates to some extent the difficulties associated with the existing techniques. Since a `good' choice of the initial iteration gives  reduced
errors,  the special method for the determination of the initial iterations has
been considered.

The error analysis of the proposed method has been provided.

Finally, the advantages of the proposed method have been illustrated with numerical experiments carried out on the basis of simulations with estimates of the covariance matrices. It has been shown, in particular, that the errors associated with the proposed technique are smaller than those associated with the existing methods. This is because of the special features of our method described above.

\section{Appendix}\label{appendix}

\subsection{Convergence }\label{convergence}

Convergence of the method presented in Section \ref{structure} can be shown on the basis of the results presented in \cite{chen87,Zhening2015} as follows.

 We call ${\bfF}=\left(F_1,...,F_p\right)\in\rt_{r_1,\ldots,r_p}$ a point in the space $\rt_{r_1,\ldots,r_p}.$ For every point  $\bfF\in\rt_{r_1,\ldots,r_p}$, define a set
 $$
 \rt_{r_j}^{\mbox{\scriptsize\boldmath $F$}} = \left\{\left(F_1,\ldots,F_{j-1}\right)\right\} \times  \rt_{r_j} \times \left\{\left(F_{j+1},\ldots,F_{p}\right) \right\},\quad \mbox{for $j=1,\ldots,p$}.
 $$
A coordinate-wise minimum point of the procedure represented by Algorithm 1 is denoted by
  ${\bfF}^*=\left(F_1^*,...,F_p^*\right)$ where\footnote{The RHS in (\ref{fjarg}) is a set since the solution of problem $\displaystyle\min_{F_j\in {\mathbb R}_{r_j}}\;f(F_1^*,...,F_{j-1}^*,F_j,F_{j+1}^*,...,F_p^*) $  is not unique.}
\begin{eqnarray}\label{fjarg}
F_j^*\in \left\{\arg\min_{F_j\in {\mathbb R}_{r_j}}\;f(F_1^*,...,F_{j-1}^*,F_j,F_{j+1}^*,...,F_p^*) \right\}.
\end{eqnarray}
This point is a local minimum of objective function in (\ref{f1ph1p}), $f({\bfF}) = \left\|H - \displaystyle\sum_{j=1}^p F_jG_j \right\|^2$.\footnote{There could be other local minimums defined differently from that in (\ref{fjarg}).}
Note that ${\bfF}^{(q+1)}$ in Algorithm 1 and $\bfF^*$ defined by (\ref{fjarg}) are, of course, different.

For $\bfF^{(q)}$ defined by Algorithm 1, denote
\begin{eqnarray}\label{bfinfty}
\check{\bfF} =\lim_{q\rightarrow \infty} \bfF^{(q)}.
\end{eqnarray}
Note that because of (\ref{eq6}), the sequence $\{\bfF^{(q)} \}$  is bounded.

\begin{theorem}\label{th}
Point  $\check{\bfF} $ defined by (\ref{bfinfty}) is the coordinate-wise minimum of Algorithm 1.
\end{theorem}

\begin{IEEEproof}
For each fixed ${\bfF}=(F_1,...,F_p)$,  a so-called best response matrix to matrix $F_j$ is denoted by  $\po_j^{\mbox{\scriptsize\boldmath $F$}}$, where
$$\po_j^{\mbox{\scriptsize\boldmath $F$}}\in\left\{\arg\min_{F_j\in{\mathbb R}_{r_j}}\;f(F_1,...F_{j-1},F_j,F_{j+1},...,F_p)\right\}.$$
Let $\{{\bfF}^{(q)}\}$ be a sequence generated by Algorithm 1, where ${\bfF}^{(q)}=(F_1^{(q)},...,F_p^{(q)})$.
 Since each $\rt_{r_j}$ is  closed  \cite[p. 304]{tu2007introduction}, there is a subsequence $\{{\bfF}^{(q_s)}\}$ such that $(F_1^{(q_s)},...,F_p^{(q_s)})\rightarrow(F_1^{*},...,F_p^{*})={\bfF}^*$ as $s\rightarrow\infty$. Then, for any $j=1,...,p$, we have
\begin{eqnarray*}
f(F_1^{(q_s)},...,F_{j-1}^{(q_s)},\po_j^{{\mbox{\scriptsize\boldmath $F$}}^*},F_{j+1}^{(q_s)},...,F_p^{(q_s)})& \geq & f(F_1^{(q_s)},...,F_{j-1}^{(q_s)},\po_j^{{
{\mbox{\scriptsize\boldmath $F$}}}^{(q_s)}},F_{j+1}^{(q_s)},...,F_p^{(q_s)})\\
                                                                                & \geq & f(F_1^{(q_s+1)},...,F_{j-1}^{(q_s+1)},F_{j}^{(q_s+1)},F_{j+1}^{(q_s+1)},...,F_p^{(q_s+1)})\\
                                                                                & \geq & f(F_1^{(q_{s+1})},...,F_{j-1}^{(q_{s+1})},F_{j}^{(q_{s+1})},F_{j+1}^{(q_{s+1})},...,F_p^{(q_{s+1})})
\end{eqnarray*}
By continuity, when $s\rightarrow\infty$,
$$f(F_1^{*},...,F_{j-1}^{*},\po_j^{{\mbox{\scriptsize\boldmath $F$}}^*},F_{j+1}^{*},...,F_p^{*})\geq f(F_1^{*},...,F_{j-1}^{*},F_{j}^{*},F_{j+1}^{*},...,F_p^{*}),$$
which implies that above should hold as an equality, since the inequality is true by the definition of the best response matrix $\po_j^{{\mbox{\scriptsize\boldmath $F$}}^*}$. Thus, $F_j^*$ is such as in (\ref{fjarg}), i.e.  $F_j^*$ is a  solution of the problem
$$\min_{F_j\in\mathcal{R}_{r_j}}\;f(F_1^*,...,F_{j-1}^*,F_j,F_{j+1}^*,...,F_p^*),\;\forall j=1,...,p.$$
\end{IEEEproof}

\begin{remark}
Theorem \ref{th} still holds if the objective function is defined as  $f(\widetilde{\bfF}) = \left\|\widetilde{H} - \displaystyle\sum_{j=1}^p \widetilde{F}_j\widetilde{G}_j \right\|^2$ where $\widetilde{\bfF}=(\widetilde{F}_1,\ldots,\widetilde{F}_p)$, and $\widetilde{H}$ and $\widetilde{G}_j$ are defined by Remark \ref{rem2}. In this case, the coordinate-wise minimum point is defined similar to that in (\ref{fjarg}) where symbols ${\bfF}^*$, $F_j^*$ and $F_j$ should be replaced with $\widetilde{\bfF}^{*}$, $\widetilde{F}_j^{*}$ and $\widetilde{F}_j$, respectively. More precisely, if    $\widetilde{\bfF}^{(q)}$ denotes the $q$th iteration of Algorithm 1 as described in Remark \ref{rem2}, then the following is true.

\begin{corollary}
Point  $\widehat{\bfF} $ defined by
$$
\widehat{\bfF} =\lim_{q\rightarrow \infty} \widetilde{\bfF}^{(q)}
$$
is the coordinate-wise minimum of Algorithm 1 in the case when covariance matrices $E_{xy}$ and $E_{yy}$ are replaced with their
estimates $\widetilde{E}_{xy}$ and $\widetilde{E}_{yy}$, respectively.
\end{corollary}
\end{remark}

\bibliographystyle{elsarticle-num}

\bibliography{references}

\begin{thebibliography}{10}
\expandafter\ifx\csname url\endcsname\relax
  \def\url#1{\texttt{#1}}\fi
\expandafter\ifx\csname urlprefix\endcsname\relax\def\urlprefix{URL }\fi
\expandafter\ifx\csname href\endcsname\relax
  \def\href#1#2{#2} \def\path#1{#1}\fi

\bibitem{chen87}
B.~Chen, S.~He, Z.~Li, S.~Zhang, Maximum block improvement and polynomial
  optimization, SIAM Journal on Optimization 22~(1) (2012) 87--107.

\bibitem{Zhening2015}
Z.~Li, A.~Uschmajew, S.~Zhang, On convergence of the maximum block improvement
  method, SIAM Journal on Optimization 25~(1) (2015) 210--233.

\bibitem{torbook2007}
A.~Torokhti, P.~Howlett, Computational Methods for Modelling of Nonlinear
  Systems, Elsevier, 2007.

\bibitem{Torokhti2007}
S.~Friedland, A.~Torokhti, Generalized rank-constrained matrix approximations,
  SIAM Journal on Matrix Analysis and Applications 29~(2) (2007) 656--659.

\bibitem{tor5277}
A.~Torokhti, S.~Friedland, Towards theory of generic {P}rincipal {C}omponent
  {A}nalysis, Journal of Multivariate Analysis 100~(4) (2009) 661 -- 669.

\bibitem{dragotti2009}
P.~L. Dragotti, M.~Gastpar, Distributed Source Coding: Theory, Algorithms and
  Applications, Academic Press, 2009.

\bibitem{5447742}
J.~Fang, H.~Li, Optimal/near-optimal dimensionality reduction for distributed
  estimation in homogeneous and certain inhomogeneous scenarios, IEEE
  Transactions on Signal Processing 58~(8) (2010) 4339--4353.

\bibitem{5504834}
A.~Amar, A.~Leshem, M.~Gastpar, Recursive implementation of the distributed
  {K}arhunen-{L}oeve transform, IEEE Transactions on Signal Processing 58~(10)
  (2010) 5320--5330.

\bibitem{Bert2010_1}
A.~Bertrand, M.~Moonen, Distributed adaptive node-specific signal estimation in
  fully connected sensor networks—part i: Sequential node updating, IEEE
  Transactions on Signal Processing 58~(10) (2010) 5277 -- 5291.

\bibitem{6334203}
M.~Lara, B.~Mulgrew, Performance of the distributed {KLT} and its approximate
  implementation, in: 2012 Proceedings of the 20th European Signal Processing
  Conference (EUSIPCO), 2012, pp. 724--728.

\bibitem{ma3098}
H.~Ma, Y.-H. Yang, Y.~Chen, K.~Liu, Q.~Wang, Distributed state estimation with
  dimension reduction preprocessing, Signal Processing, IEEE Transactions on
  62~(12) (2014) 3098--3110.

\bibitem{Marelli201527}
D.~E. Marelli, M.~Fu, Distributed weighted least-squares estimation with fast
  convergence for large-scale systems, Automatica 51~(0) (2015) 27 -- 39.

\bibitem{Scharf1991113}
L.~L. Scharf, The $\text{{SVD}}$ and reduced rank signal processing, Signal
  Processing 25~(2) (1991) 113 -- 133.

\bibitem{681430}
Y.~Hua, W.~Liu, Generalized {K}arhunen-{L}oeve transform, IEEE Signal
  Processing Letters 5~(6) (1998) 141--142.

\bibitem{905856}
Y.~Hua, M.~Nikpour, P.~Stoica, Optimal reduced-rank estimation and filtering,
  IEEE Transactions on Signal Processing 49~(3) (2001) 457--469.

\bibitem{Torokhti20102822}
A.~Torokhti, S.~Miklavcic, Data compression under constraints of causality and
  variable finite memory, Signal Processing 90~(10) (2010) 2822 -- 2834.

\bibitem{4276987}
I.~D. Schizas, G.~B. Giannakis, Z.-Q. Luo, Distributed estimation using
  reduced-dimensionality sensor observations, IEEE Transactions on Signal
  Processing 55~(8) (2007) 4284--4299.

\bibitem{2001309}
A.~Torokhti, P.~Howlett, Optimal fixed rank transform of the second degree,
  IEEE Trans. CAS. Part II, Analog and Digital Signal Processing 48~(3) (2001)
  309 -- 315.

\bibitem{1055037}
D.~Slepian, J.~Wolf, Noiseless coding of correlated information sources, IEEE
  Transactions on Information Theory 19~(4) (1973) 471--480.

\bibitem{1055508}
A.~Wyner, J.~Ziv, The rate-distortion function for source coding with side
  information at the decoder, IEEE Transactions on Information Theory 22~(1)
  (1976) 1--10.

\bibitem{Song20052131}
E.~Song, Y.~Zhu, J.~Zhou, Sensors optimal dimensionality compression matrix in
  estimation fusion, Automatica 41~(12) (2005) 2131 -- 2139.

\bibitem{1420805}
Y.~Zhu, E.~Song, J.~Zhou, Z.~You, Optimal dimensionality reduction of sensor
  data in multisensor estimation fusion, IEEE Transactions on Signal Processing
  53~(5) (2005) 1631--1639.

\bibitem{4016296}
M.~Gastpar, P.~Dragotti, M.~Vetterli, The distributed {K}arhunen-{L}o\`{e}ve
  transform, IEEE Transactions on Information Theory 52~(12) (2006) 5177--5196.

\bibitem{4475387}
O.~Roy, M.~Vetterli, Dimensionality reduction for distributed estimation in the
  infinite dimensional regime, IEEE Transactions on Information Theory 54~(4)
  (2008) 1655--1669.

\bibitem{952802}
V.~Goyal, Theoretical foundations of transform coding, IEEE Signal Processing
  Magazine 18~(5) (2001) 9--21.

\bibitem{Schizas2007}
I.~D. Schizas, A.~Ribeiro, G.~B. Giannakis, Dimensionality reduction,
  compression and quantization for distributed estimation with wireless sensor
  networks, in: Wireless Communications, Vol. 143 of The IMA Volumes in
  Mathematics and its Applications, Springer New York, 2007, pp. 259--296.

\bibitem{Saghri2010}
J.~A. Saghri, S.~Schroeder, A.~G. Tescher, Adaptive two-stage
  {K}arhunen-{L}oeve-transform scheme for spectral decorrelation in
  hyperspectral bandwidth compression, Optical Engineering 49~(5) (2010)
  057001--057001--7.

\bibitem{Torokhti2012}
V.~Ejov, A.~Torokhti, How to transform matrices ${U}_1$,...,${U}_p$ to matrices
  ${V}_1$,...,${V}_p$ so that ${V}_i{V}_j=\mathbb{O}$ if $i\neq j$?, Numerical
  Algebra, Control and Optimization 2~(2) (2012) 293--299.

\bibitem{Brillinger2001}
D.~R. Brillinger, Time Series: Data Analysis and Theory, Holden Day, San
  Francisco, 2001.

\bibitem{Torokhti2009661}
A.~Torokhti, S.~Friedland, Towards theory of generic principal component
  analysis, Journal of Multivariate Analysis 100~(4) (2009) 661 -- 669.

\bibitem{Tseng2001}
P.~Tseng, Convergence of a block coordinate descent method for
  nondifferentiable minimization, Journal of Optimization Theory and
  Applications 109~(3) (2001) 475--494.

\bibitem{bertsekas1995nonlinear}
D.~Bertsekas, Nonlinear Programming, Athena Scientific, 1995.

\bibitem{149980}
L.~Perlovsky, T.~Marzetta, Estimating a covariance matrix from incomplete
  realizations of a random vector, IEEE Transactions on Signal Processing
  40~(8) (1992) 2097--2100.

\bibitem{Ledoit2004365}
O.~Ledoit, M.~Wolf, A well-conditioned estimator for large-dimensional
  covariance matrices, Journal of Multivariate Analysis 88~(2) (2004) 365 --
  411.

\bibitem{ledoit2012}
O.~Ledoit, M.~Wolf, Nonlinear shrinkage estimation of large-dimensional
  covariance matrices, Ann. Statist. 40~(2) (2012) 1024--1060.

\bibitem{Adamczak2009}
R.~Adamczak, A.~E. Litvak, A.~Pajor, N.~Tomczak-Jaegermann, Quantitative
  estimates of the convergence of the empirical covariance matrix in
  log-concave ensembles, Journal of the American Mathematical Society~(23)
  (2009) 535--561.

\bibitem{Vershynin2012}
R.~Vershynin, How close is the sample covariance matrix to the actual
  covariance matrix?, Journal of Theoretical Probability 25~(3) (2012)
  655--686.

\bibitem{won2013}
S.-J.~K. Joong-Ho~Won, Johan~Lim, B.~Rajaratnam, Condition-number-regularized
  covariance estimation, Journal of the Royal Statistical Society: Series B
  (Statistical Methodology) 75~(3) (2013) 427--450.

\bibitem{schmeiser1991}
B.~W. Schmeiser, M.~H. Chen, On hit-and-run {M}onte {C}arlo sampling for
  evaluating multidimensional integrals, Technical Report 91-39, Dept.
  Statistics, Purdue Univ.

\bibitem{yang1994}
R.~Yang, J.~O. Berger, Estimation of a covariance matrix using the reference
  prior, The Anna1s of Statistics 22~(3) (1994) 1195--1211.

\bibitem{tu2007introduction}
L.~Tu, An Introduction to Manifolds, Universitext, Springer, 2007.

\end{thebibliography}


\begin{figure}
  \centering
  \includegraphics[width=7.7cm,height=6.3cm]{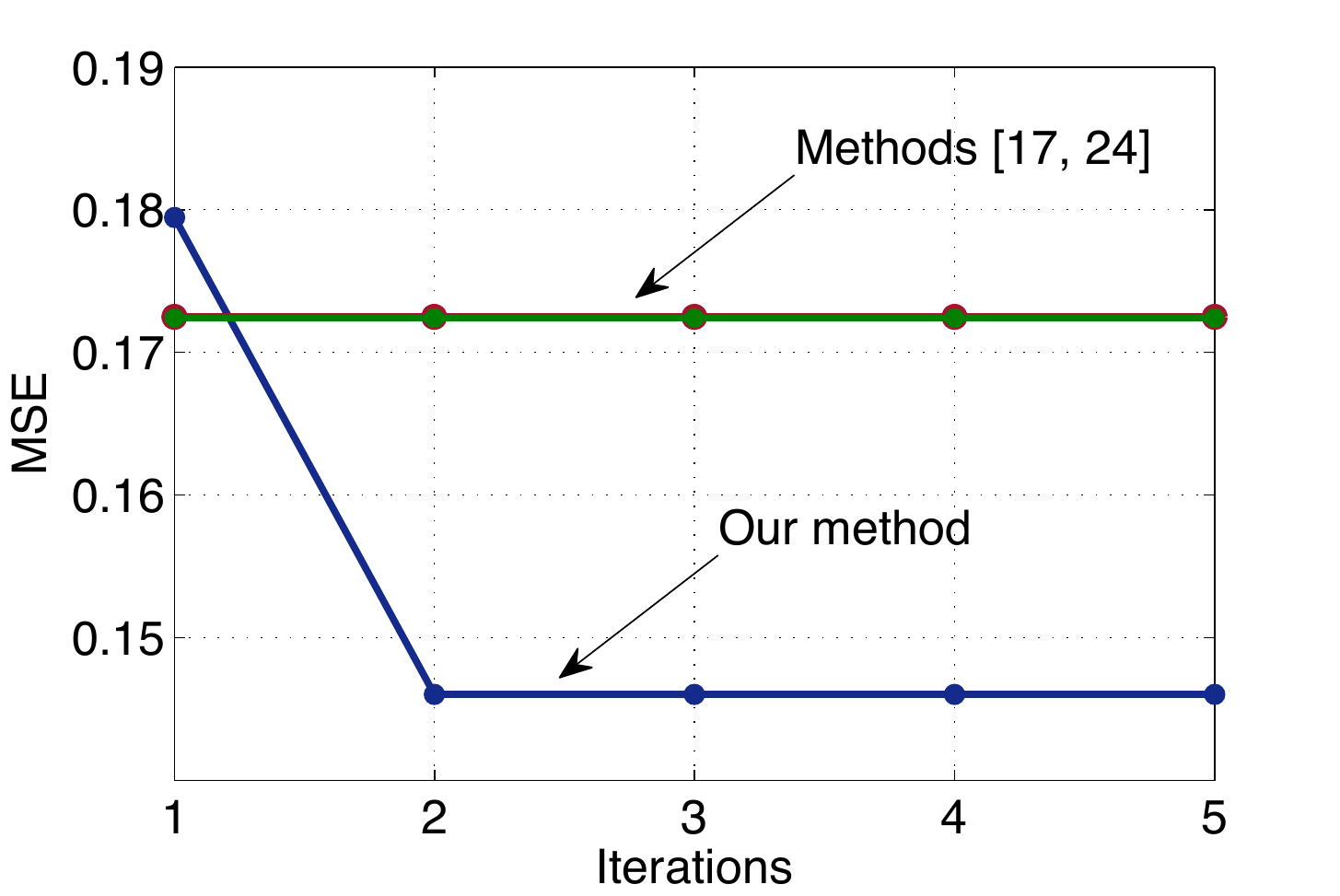}\\
  \caption{Example \ref{ex0}: Diagrams of MSE's  associated with the proposed method and the known methods versus number of iterations.}
  \label{fig01}
\end{figure}

\begin{figure}[h!]
\centering
\begin{tabular}{c@{\hspace*{5mm}}c}
\subfigure[ Example \ref{ex1}: $p=2$, $m=n_j=10$, $s=20$,   $r_j=5+j$, \newline $\sigma_j=0.2-0.1j$, for  $j=1,2,$. ]{\includegraphics[width=7.7cm,height=6.3cm]{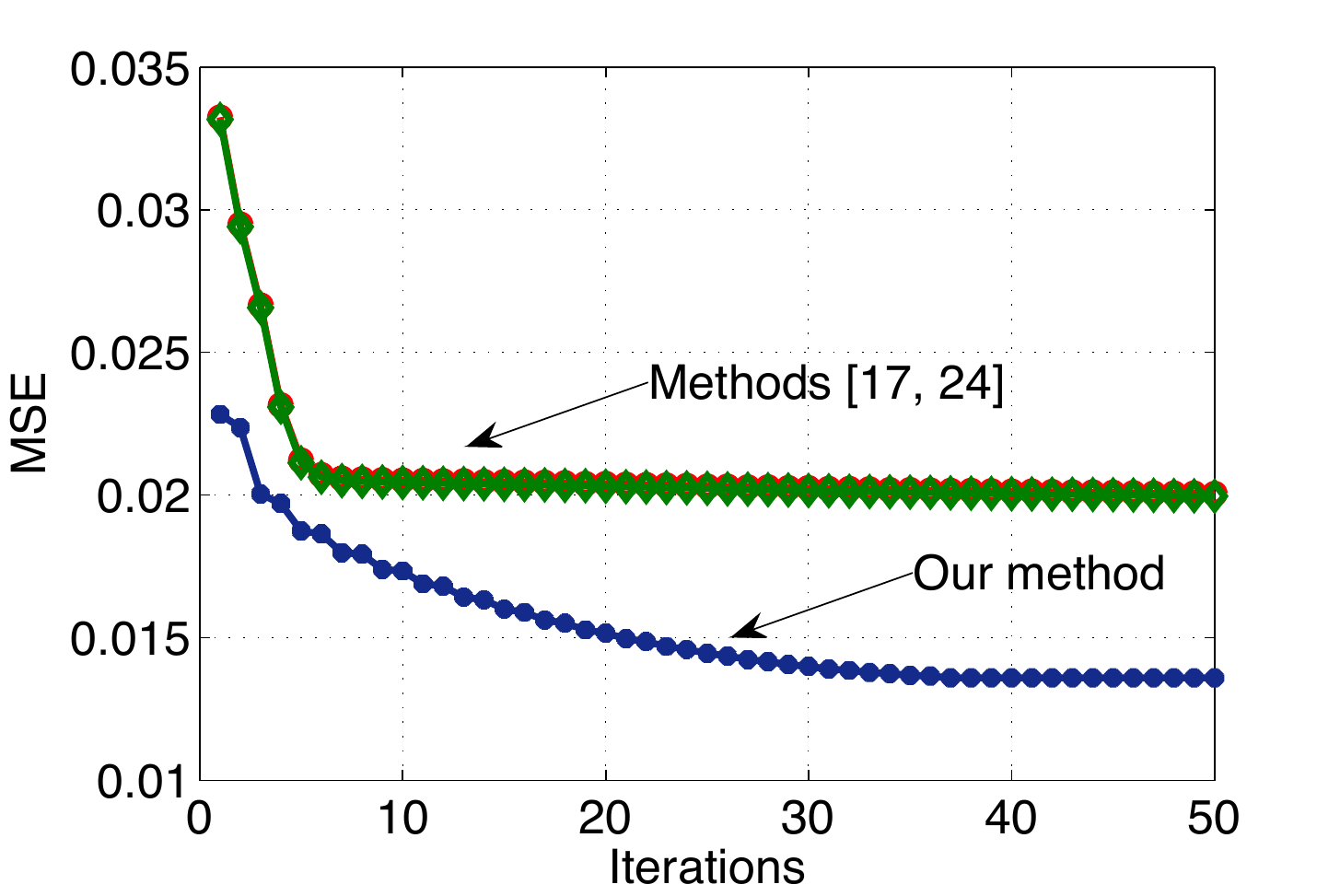}} \label{ex1:fig1a} &
\subfigure[Example \ref{ex1}: $p=3$, $m=n_j=100$, $s=10000$, $r_j=20$,
  $\sigma_j=1$, for $j=1,2,3.$]{\includegraphics[width=7.7cm,height=6.3cm]{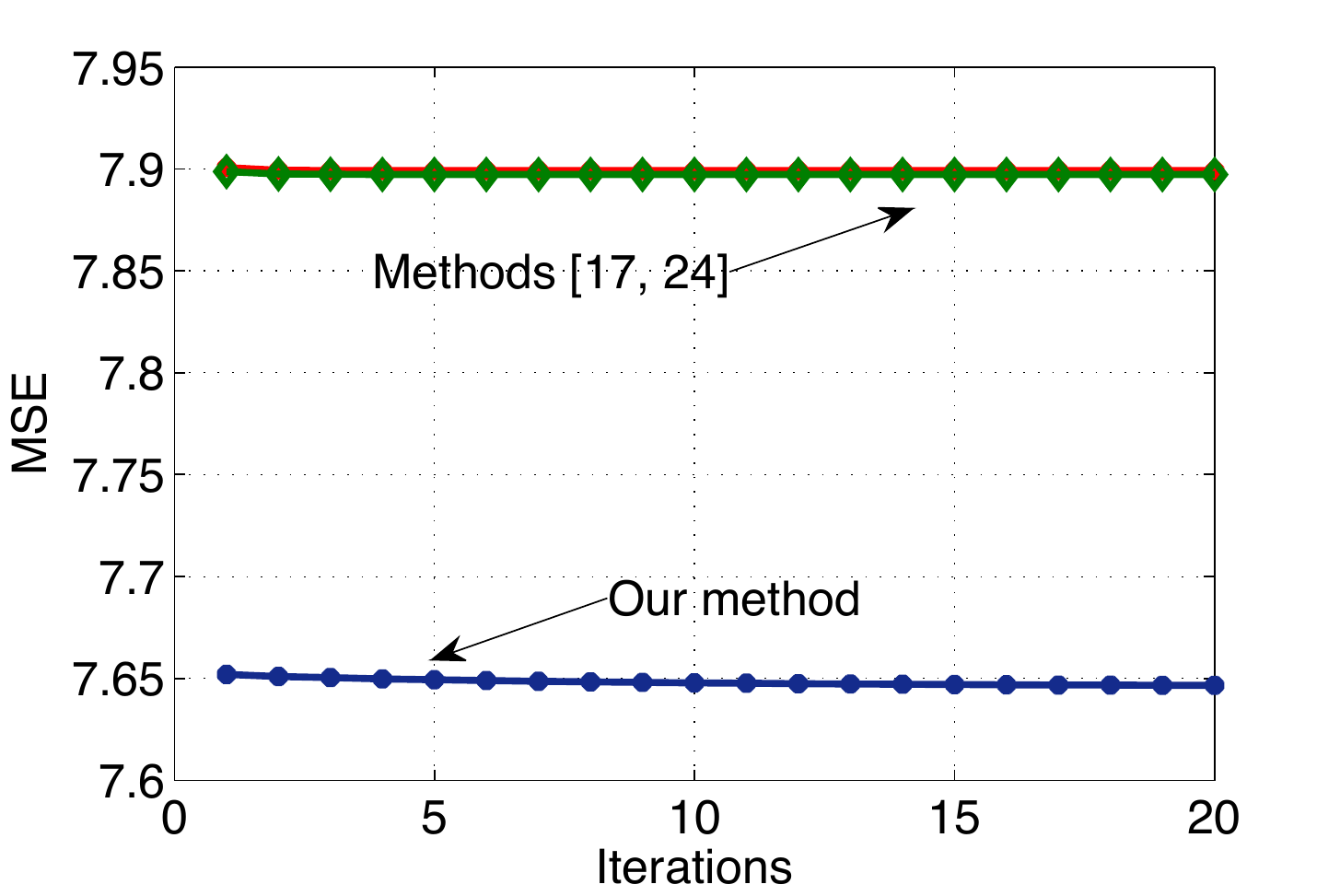}}\\
\subfigure[Example \ref{ex2}: $p=2$, $m=n_j=2$, $s=4$,  $r_j=1$, $\sigma_j=1$, for  $j=1,2$.]{\includegraphics[width=7.7cm,height=6.3cm]{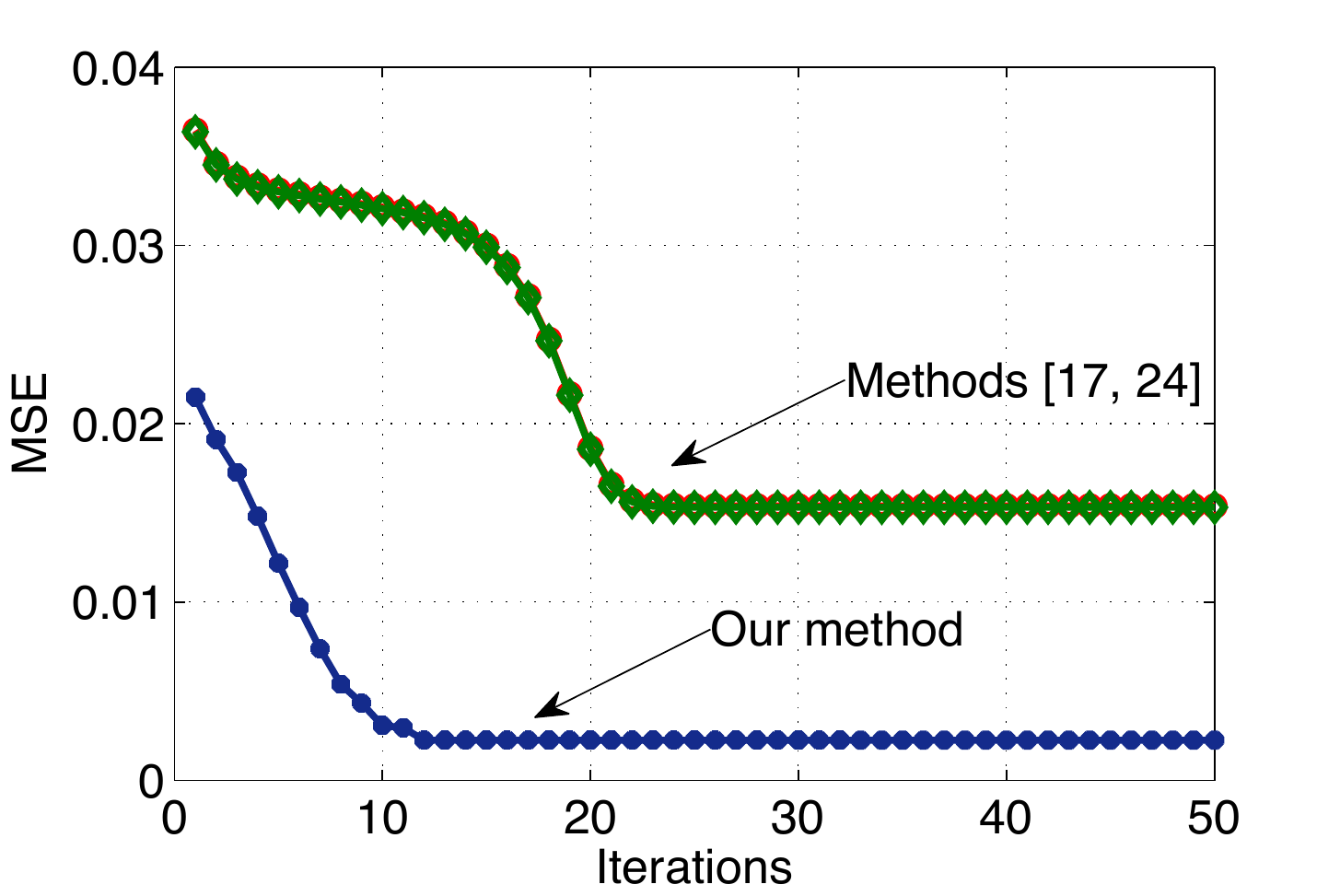}}   \label{ex2:fig3} &
\subfigure[Example \ref{ex2}: $p=2$, $m=n_j=100$, $s=250$, $r_j=20$, $\sigma_j=1$, for $j=1,2$.]{\includegraphics[width=7.7cm,height=6.3cm]{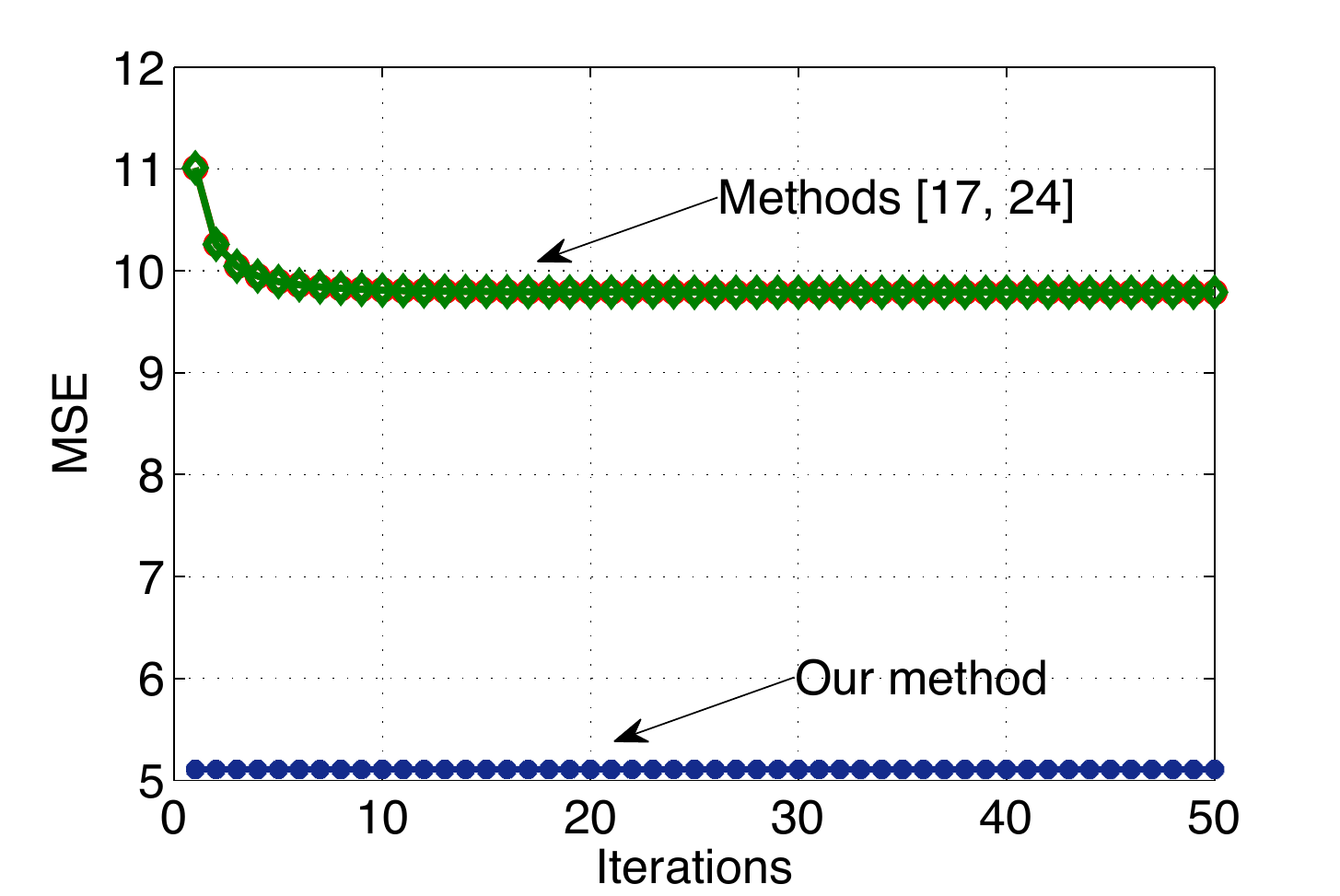}} \label{ex2:fig4a}
\end{tabular}
  \caption{Diagrams of MSE's associated with the proposed method and the known methods versus number of iterations, for $p=2$ (i.e. for two sensors) and $p=3$  (i.e. for three sensors), and different choices of signal dimensions, $m$, $n_j$, $r_j$, sample size $s$ and noise `level'  $\sigma_j$.  }
  \label{ex:fig1}
\end{figure}

\begin{figure}[h!]
\centering
\begin{tabular}{c@{\hspace*{5mm}}c}
\subfigure[Example \ref{ex2}: $p=3$, $m=n_j=100$, $s=400$, $r_j=20$, $\sigma_j=1$, for $j=1,2,3$.]{\includegraphics[,width=7.7cm,height=6.3cm]{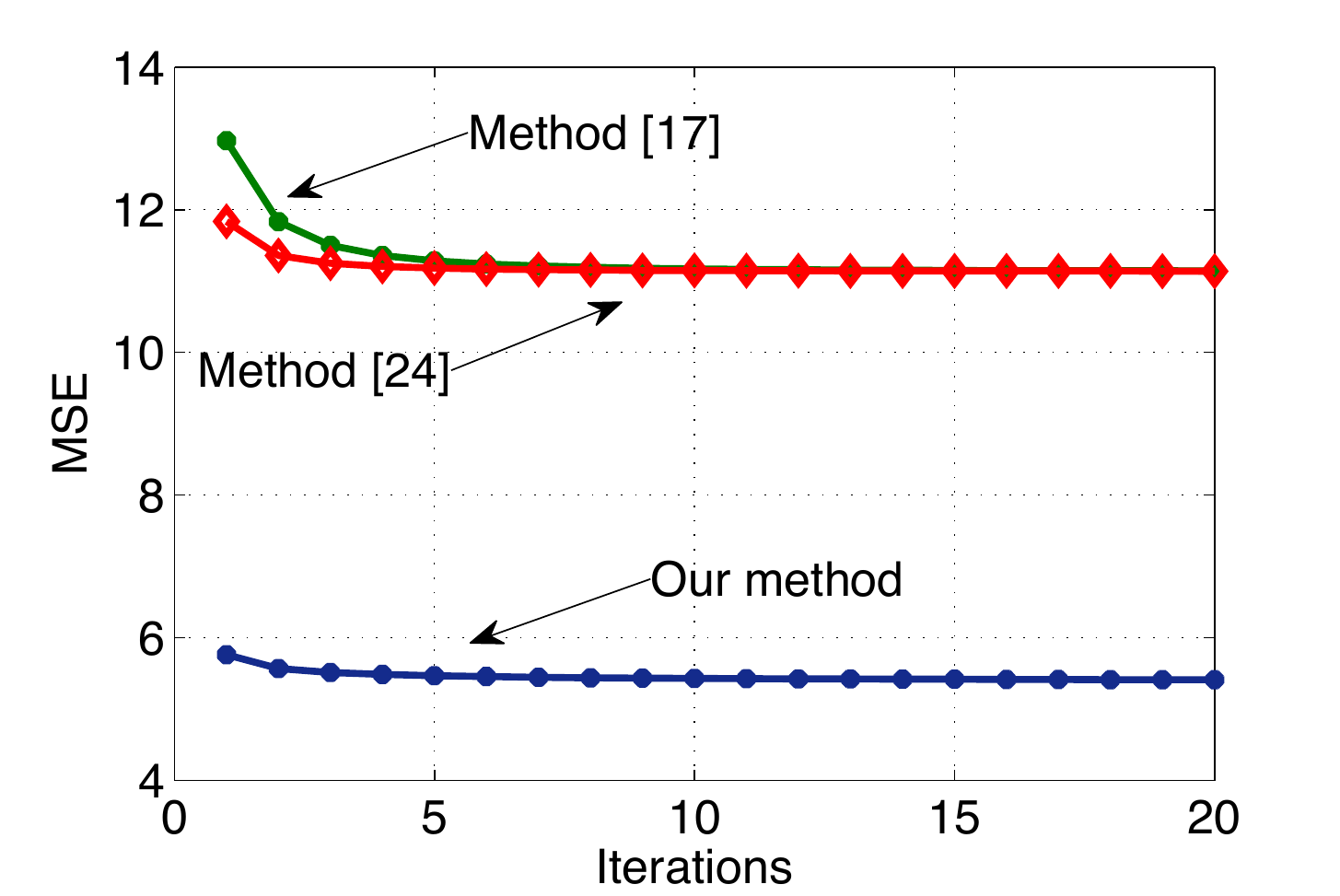}}\label{ex2:fig4b}   &
\subfigure[Example \ref{ex3}: $p=3$, $m=n_j=20$, $s=20$, $r_j=5$, $\sigma_j=0.1j$, $j=1,2,3.$]{\includegraphics[,width=7.7cm,height=6.3cm]{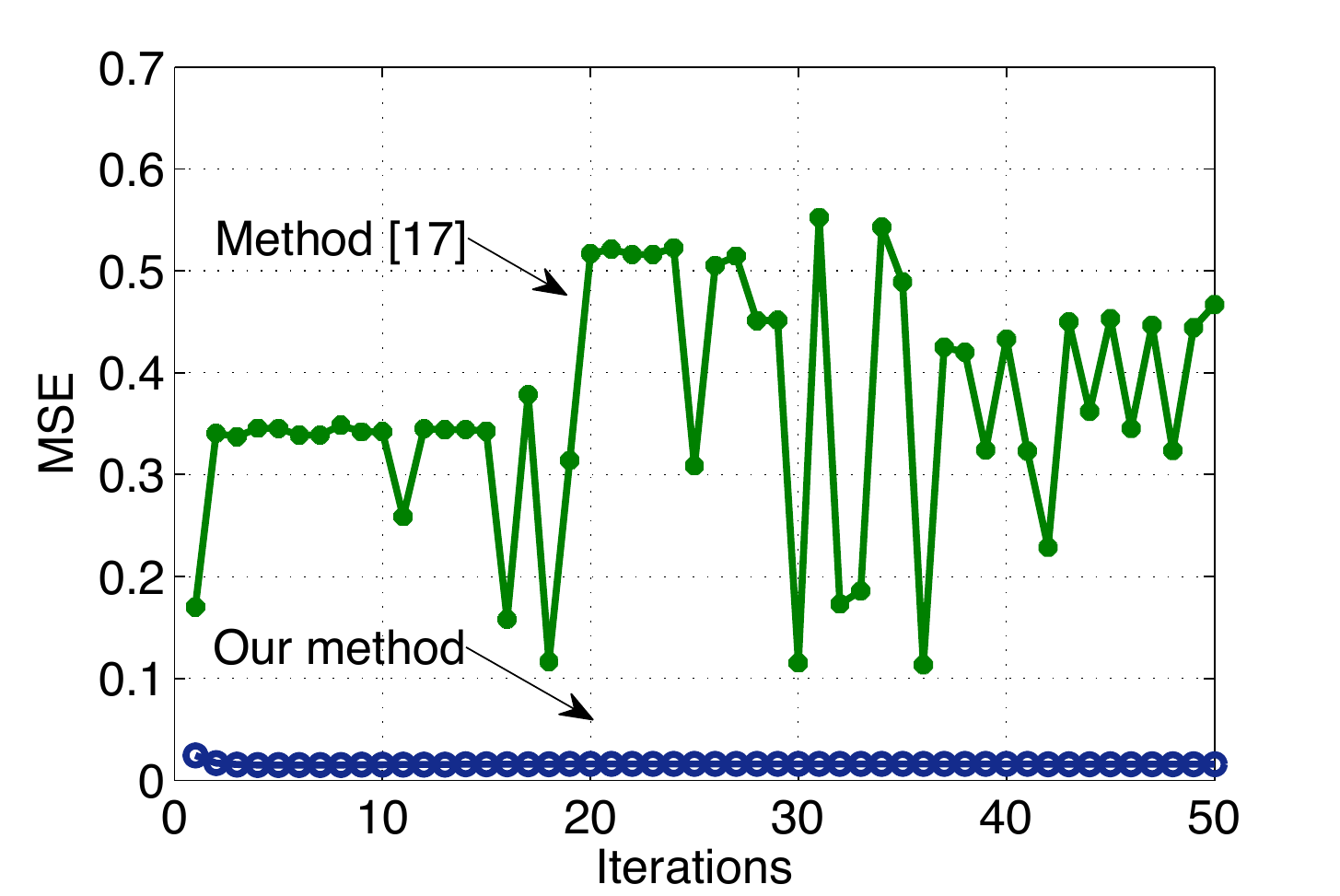}}\label{ex3:fig4a} \\
\subfigure[Example \ref{ex3}: $p=3$, $m=n_j=100$, $s=400$, $r_j=25$, $\sigma_j=0.1j,$ for $j=1,2,3$.]{\includegraphics[,width=7.7cm,height=6.3cm]{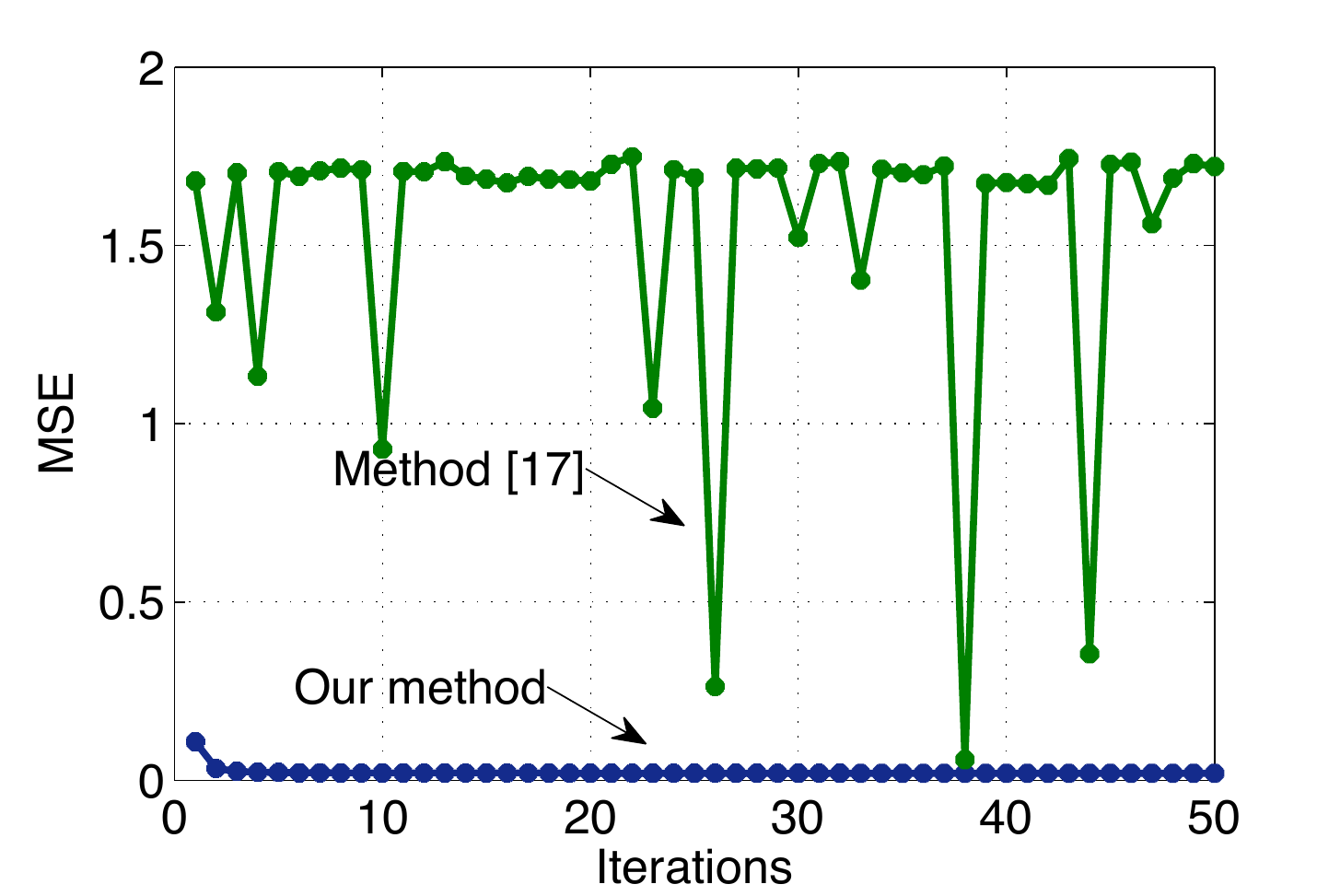}} \label{ex3:fig5b} &
\subfigure[Example \ref{ex3}: $p=2$, $m=n_j=100$, $s=50$, $r_j=20$, $\sigma_j=0.1j$,  for $j=1,2$.]{\includegraphics[width=7.7cm,height=6.3cm]{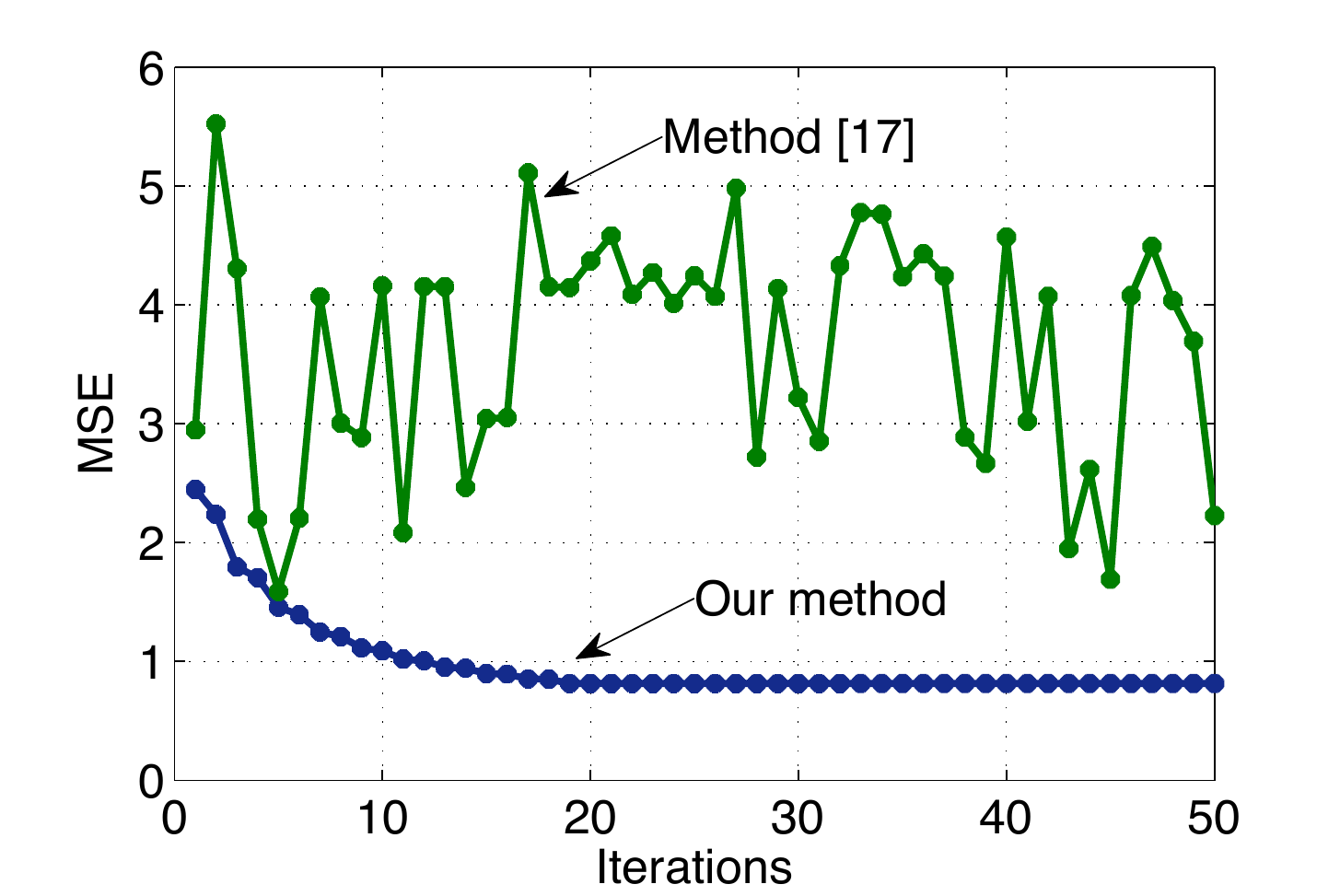}}\label{ex3:fig52}
\end{tabular}
  \caption{Diagrams of MSE's  associated with the proposed method and the known methods versus number of iterations, for $p=2$ (i.e. for two sensors) and $p=3$  (i.e. for three sensors), and different choices of signal dimensions, $m$, $n_j$, $r_j$, sample size $s$ and noise `level'  $\sigma_j$.  }
  \label{ex:fig2}
\end{figure}

\hspace*{-10mm}\begin{figure}[p]
\centering
 \vspace*{-5mm}\begin{tabular}{c@{\hspace*{5mm}}c}
\hspace*{-5mm}\vspace*{-1mm}\subfigure[Training reference signal $X.$]{\includegraphics[width=7.3cm,height=6.5cm]{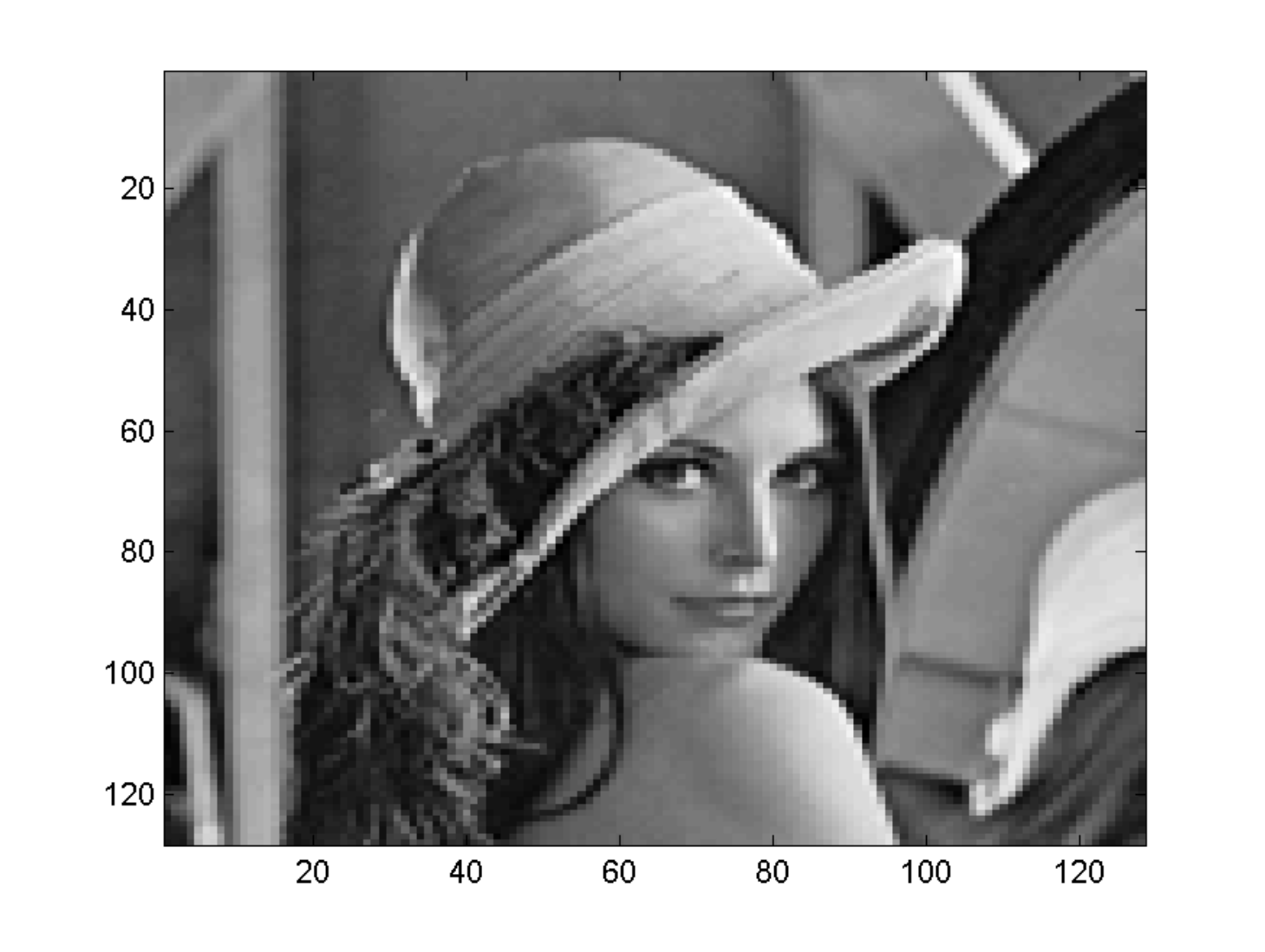}} &
\hspace*{-5mm}\vspace*{-1mm}\subfigure[Observed signal  $Y_1.$]{\includegraphics[width=7.3cm,height=6.5cm]{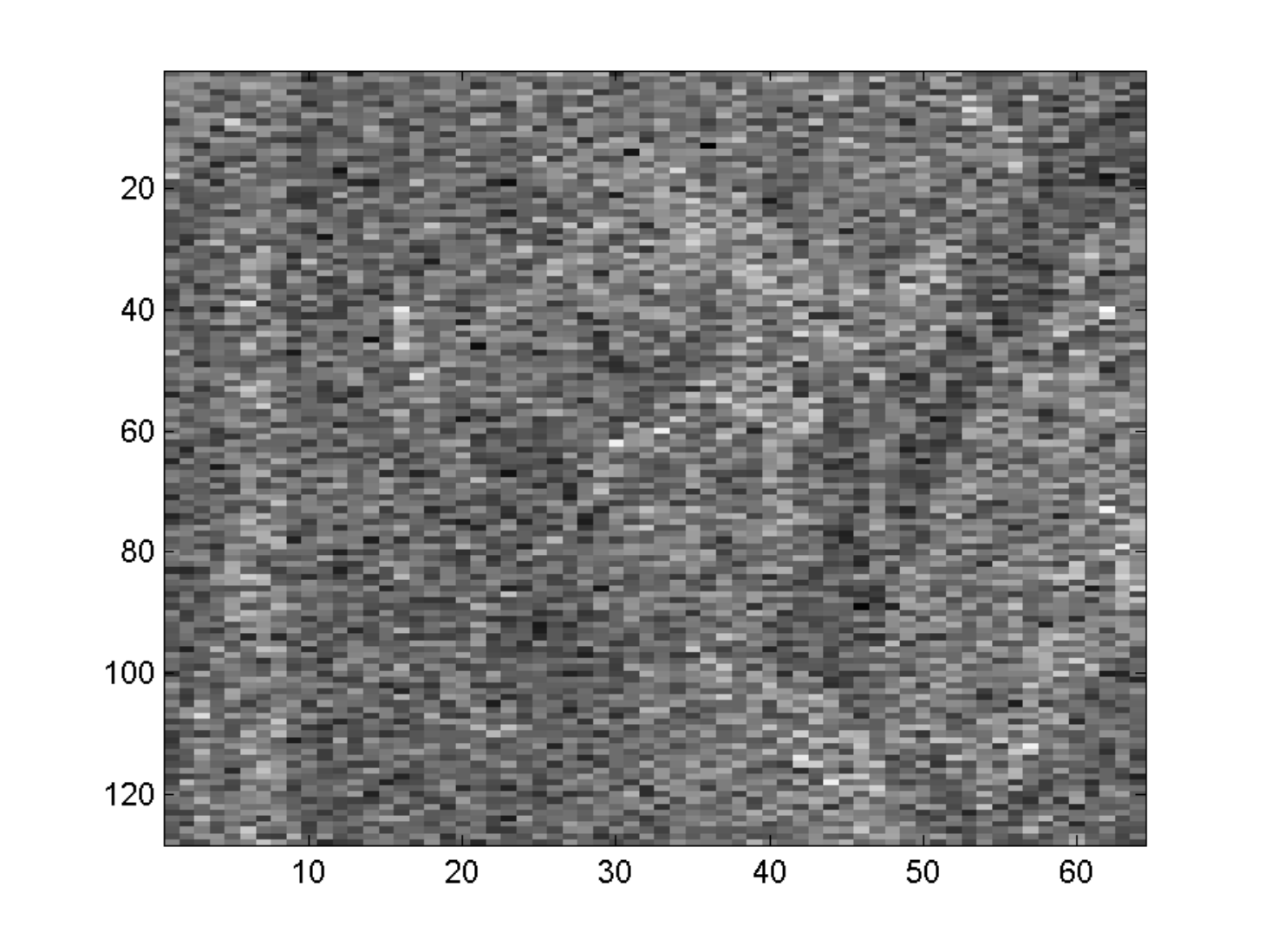}}\\
\hspace*{-5mm} \vspace*{-1mm}\subfigure[Observed signal  $Y_2.$]{\includegraphics[width=7.3cm,height=6.5cm]{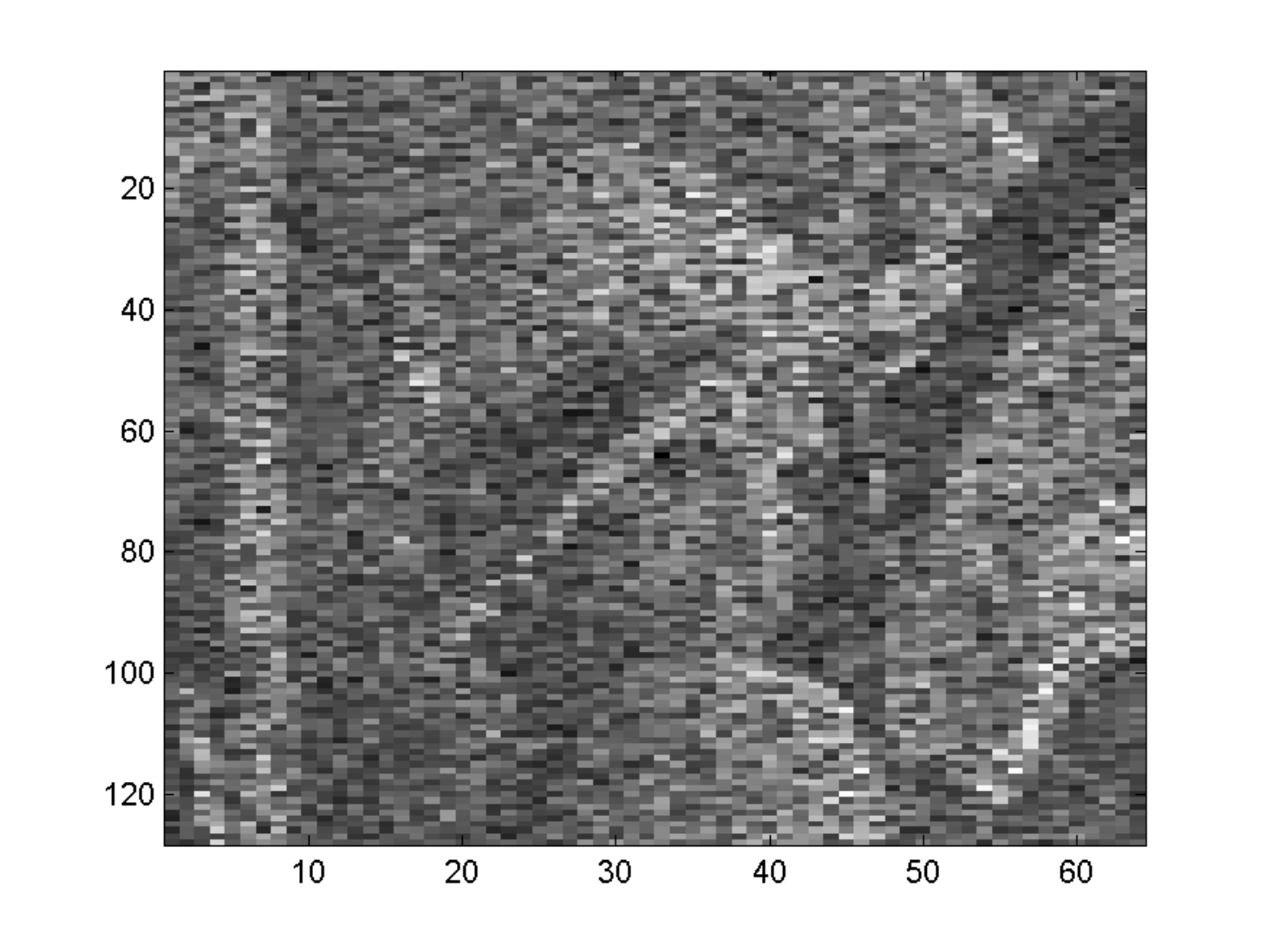}} &
\hspace*{-5mm}\vspace*{-1mm}\subfigure[Estimate of $X$ by our method.]{\includegraphics[width=7.3cm,height=6.5cm]{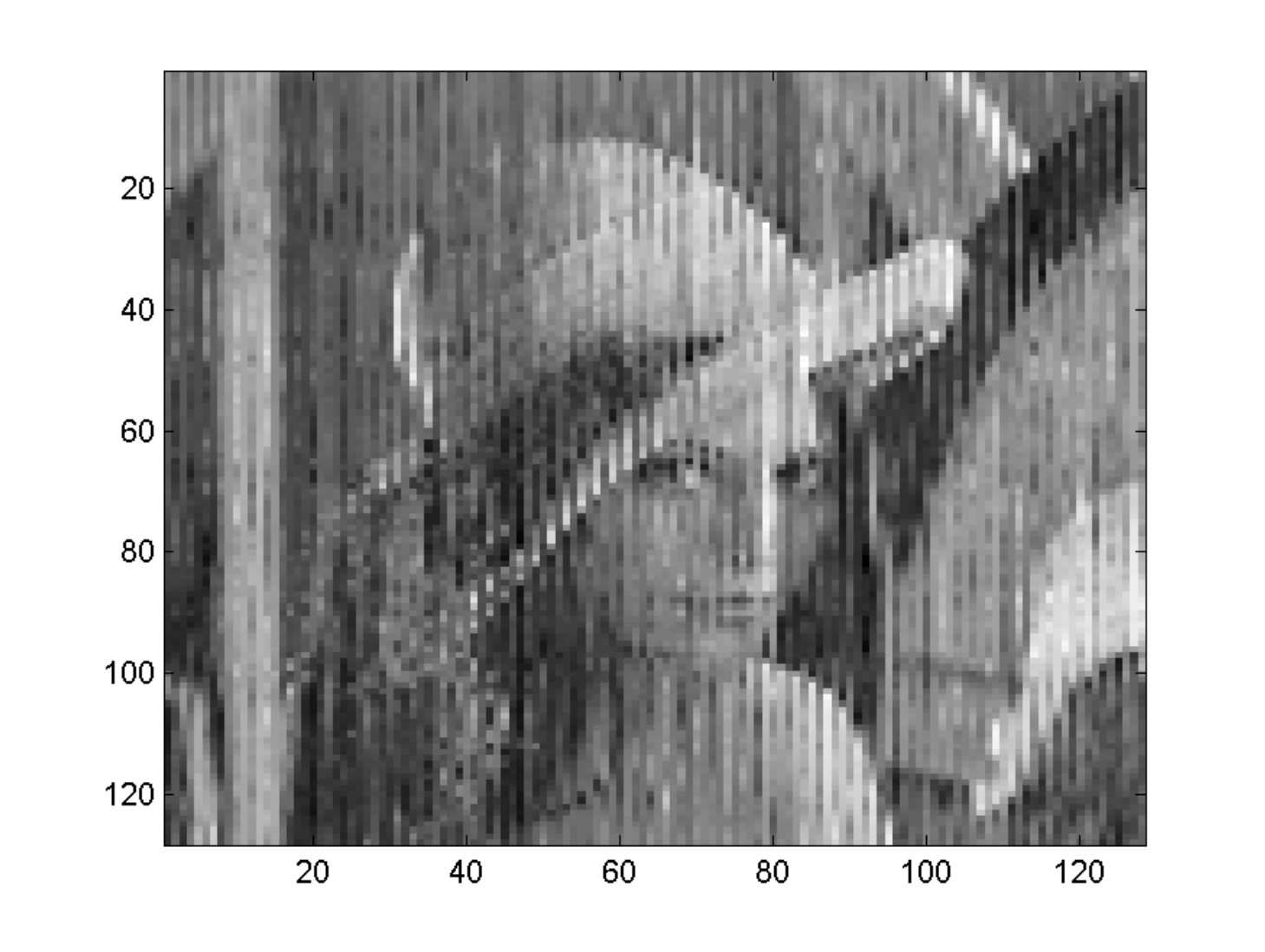}}  \\
\hspace*{-5mm}\vspace*{-1mm}\subfigure[Estimate of $X$ by method \cite{4276987}. ]{\includegraphics[width=7.3cm,height=6.5cm]{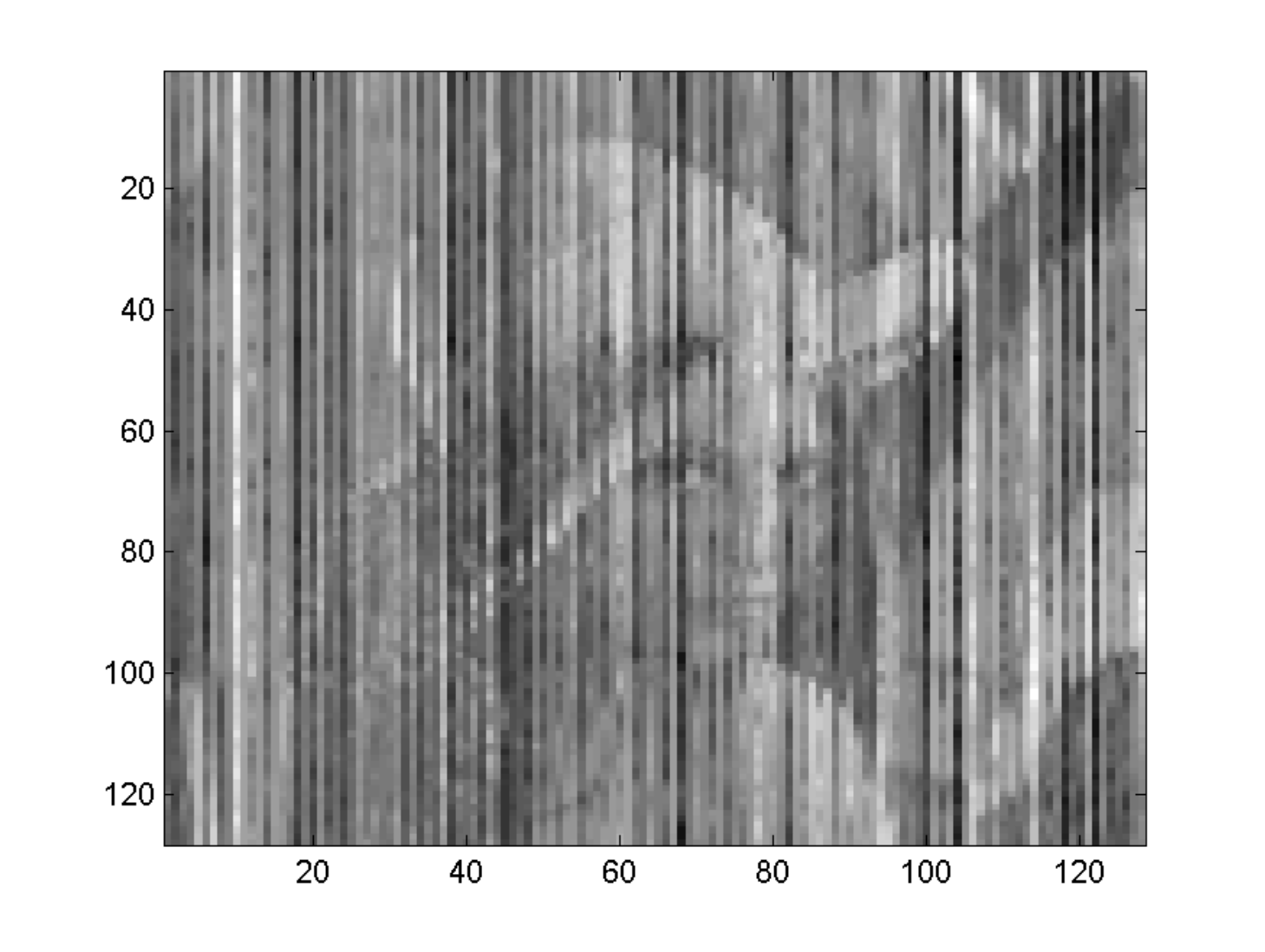}} &
\hspace*{-5mm} \vspace*{-1mm}\subfigure[Estimate of $X$ by method \cite{4475387}.]{\includegraphics[width=7.3cm,height=6.5cm]{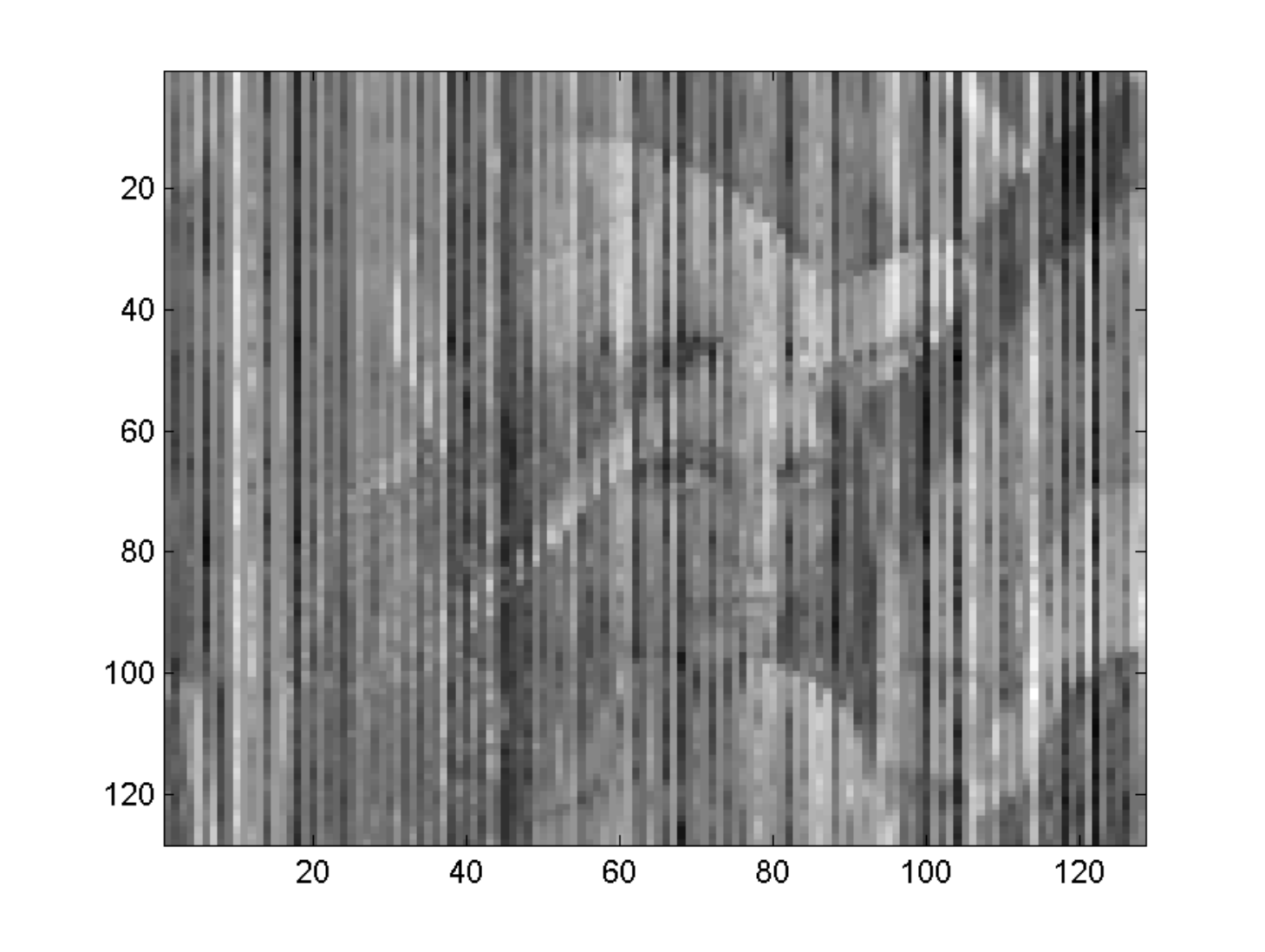}}
\end{tabular}
 \vspace*{3mm}\caption{Illustration to Example \ref{ex4}.}
 \label{ex4:fig1}
 \end{figure}

\hspace*{-10mm}\begin{figure}[p]
\centering
 \vspace*{-5mm}\begin{tabular}{c@{\hspace*{5mm}}c}
\hspace*{-5mm}\vspace*{-3mm}\subfigure[Error associated with our method.]{\includegraphics[width=12cm,height=6.5cm]{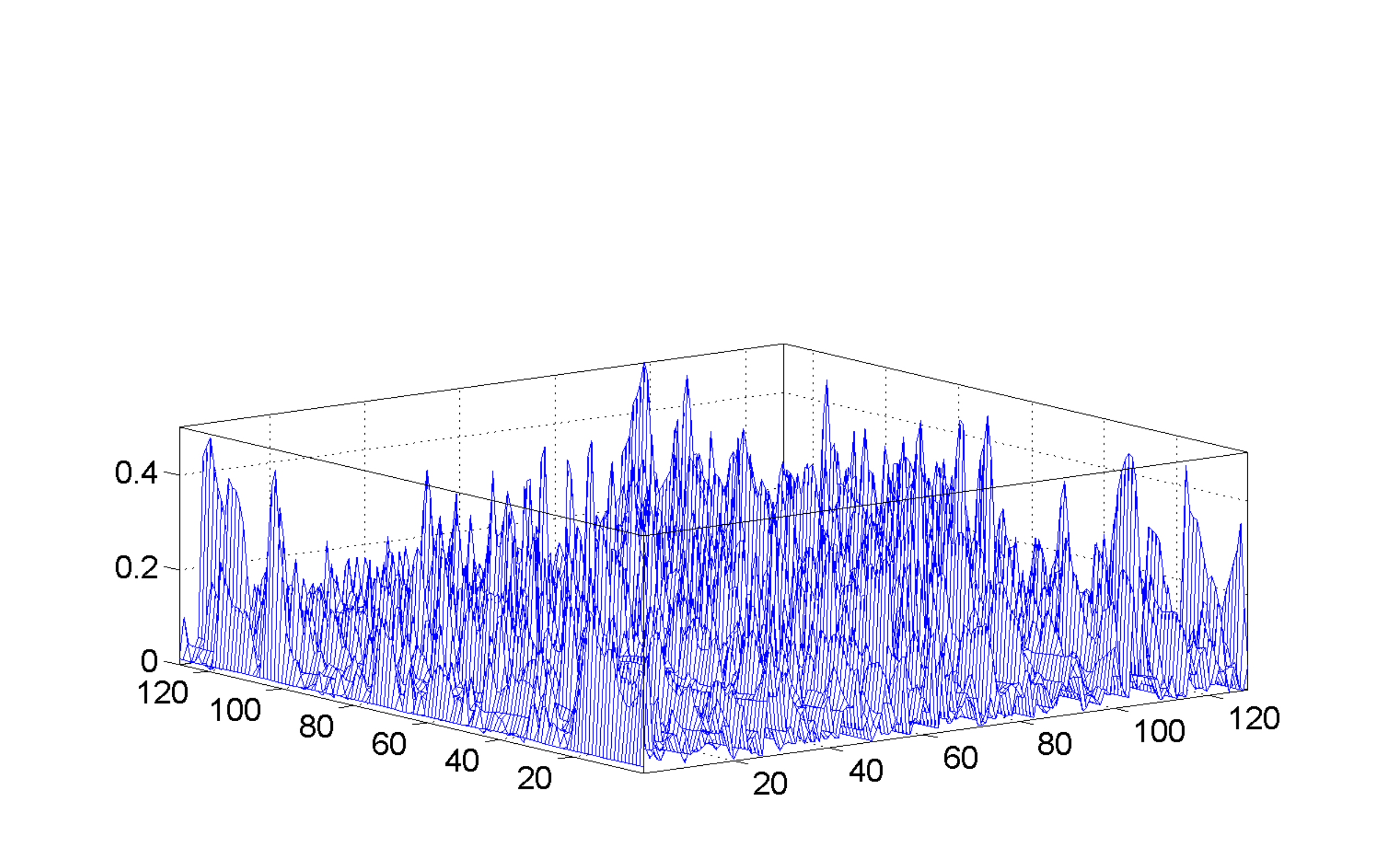}}  \\
\hspace*{-5mm}\vspace*{-3mm}\subfigure[Error associated with  method \cite{4276987}.]{\includegraphics[width=12cm,height=6.5cm]{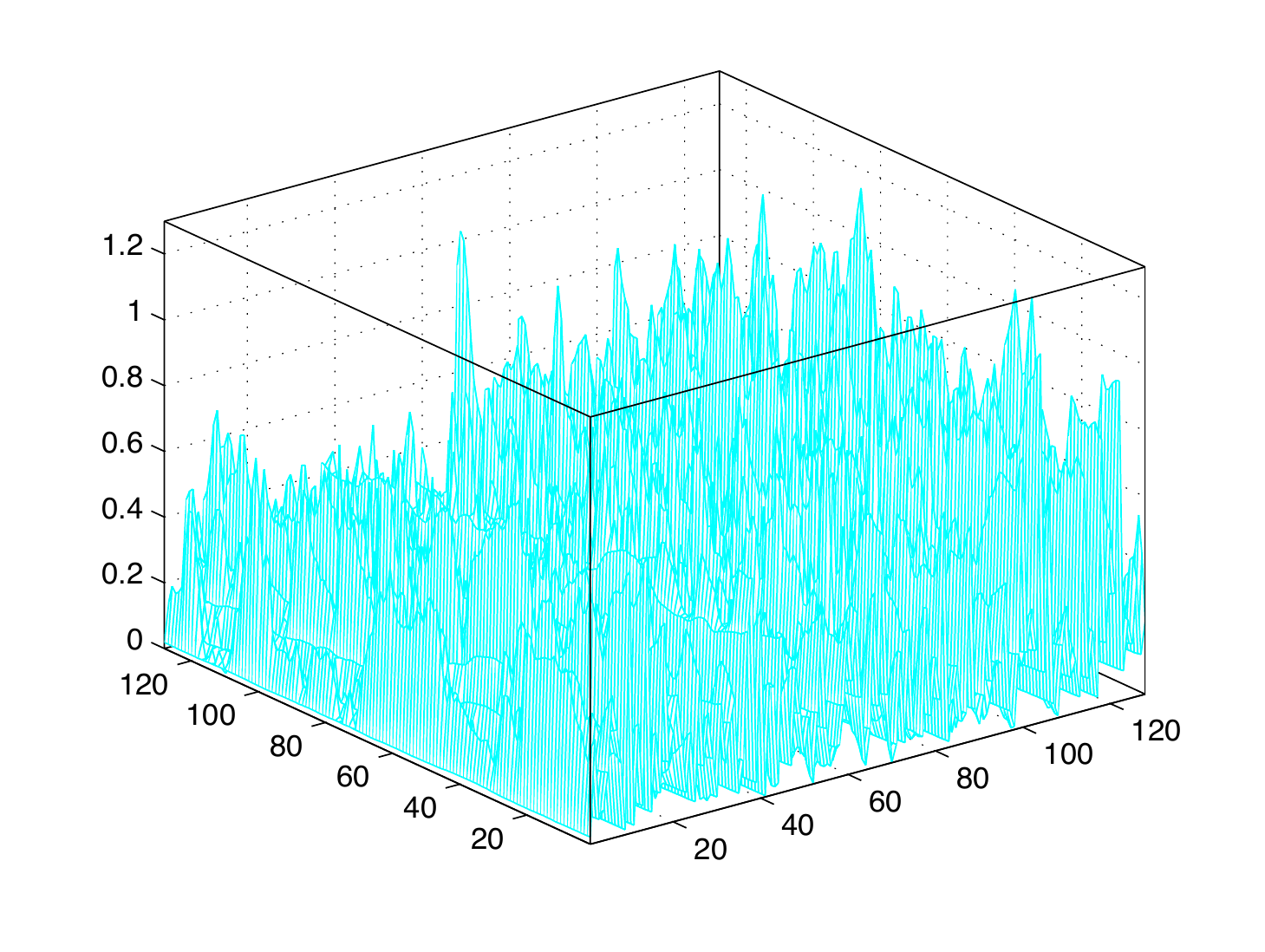}}\\
\hspace*{-5mm}\vspace*{-3mm}\subfigure[Error associated with  method \cite{4475387}. ]{\includegraphics[width=12cm,height=6.5cm]{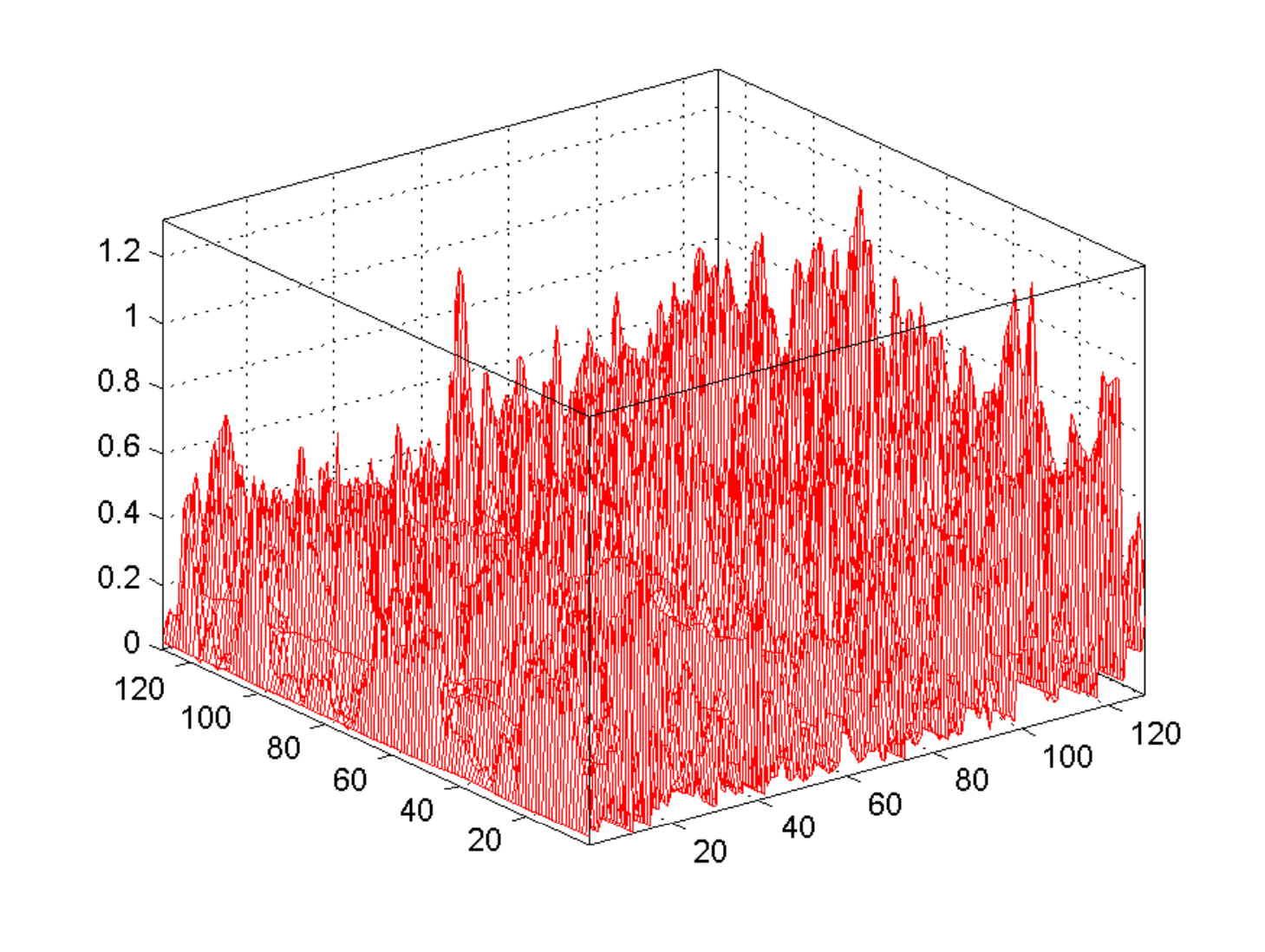}}
\end{tabular}
 \vspace*{0mm}\caption{Illustration to Example \ref{ex4}.}
 \label{ex4:fig2}
 \end{figure}

\end{document}